\documentclass[12pt]{amsart}

\usepackage{upgreek, amssymb, amsmath}
\usepackage[margin=3.5 cm]{geometry}
\usepackage[pdftex]{graphicx}
\usepackage[shortlabels]{enumitem}
\usepackage{caption,subcaption} 
\usepackage{tikz}
\usepackage{hyperref} 

\newtheorem{theorem}{Theorem}[section]
\newtheorem{fact}{Fact}[section]
\newtheorem{prop}{Proposition}[section]

\newtheorem{claim}{Claim}[section]

\newtheorem{lemma}[theorem]{Lemma}

\theoremstyle{definition}

\newtheorem{ExampleTheorem}[theorem]{Example Theorem}

\theoremstyle{remark}
\newtheorem{remark}[theorem]{Remark}

\numberwithin{equation}{section}

\newcommand{\Stefanescu}{\c{S}tef\u{a}nescu} 

\newcommand{\comments}[1]{}
\newcommand{\N}{\mathbb{N}}
\newcommand{\Z}{\mathbb{Z}}
\newcommand{\R}{\mathbb{R}}
\newcommand{\C}{\mathbb{C}}

\newcommand{\T}{\mathbb{T}}

\newcommand{\abs}[1]{\lvert#1\rvert}

\DeclareMathOperator{\im}{Im}
\newcommand{\ud}{\mathrm{d}}

\begin{document}

\title[]{Subcritical behavior for quasi-periodic Schr\"odinger cocycles with trigonometric potentials}

\author{C. A. Marx, L. H. Shou, J. L. Wellens}

\address{Department of Mathematics, Oberlin College, Oberlin, OH - 44074}
\email{cmarx@oberlin.edu}

\address{Department of Mathematics, Caltech, Pasadena, CA - 91125}
\email{lshou@caltech.edu}

\address{Department of Mathematics, MIT, Cambridge, MA 02139}
\email{jwellens@MIT.edu}

\thanks{L. H. Shou was supported by the Caltech's Summer Undergraduate Research Fellowship (SURF). J. L. Wellens was supported by the Marcella Bonsall SURF Fellowship.}





\maketitle

\begin{abstract}
We give a criterion implying subcritical behavior for quasi-periodic Schr\"odinger operators where the potential sampling function is given by a trigonometric polynomial. Subcritical behavior, in the sense of Avila's global theory, is known to imply purely absolutely continuous spectrum for all irrational frequencies and all phases.
\end{abstract}

\section{Introduction} \label{sec_intro}

Let $\mathbb{T}:=\mathbb{R}/ \mathbb{Z}$ and fix an irrational number $\alpha$, subsequently referred to as the frequency. Evaluating an {\em{analytic}} function $v:\mathbb{T}\to\R$ along the trajectories of rotation by $\alpha$ with varying starting point $x \in \mathbb{T}$ determines a quasi-periodic Schr\"odinger operator,
\begin{equation} \label{eq_qso}
(H_{v(x); \alpha}\psi)_n:=\psi_{n-1}+\psi_{n+1}+v(x+n\alpha)\psi_n ~\mbox{.}
\end{equation}
For every realization of the {\em{phase}} $x \in \mathbb{T}$, (\ref{eq_qso}) is bounded and self-adjoint on $\mathit{l}^2(\mathbb{Z})$. 

In physics, quasi-periodic Schr\"odinger operators describe the conductivity of electrons in a two-dimensional crystal layer subject to an external magnetic field of flux $\alpha$ acting perpendicular to the lattice plane. In this context the potential sampling function $v$ is usually a trigonometric polynomial, which, through its Fourier coefficients (``{\em{coupling constants}}''), carries information about the material properties of the crystal. The most well-known example is the almost Mathieu operator (AMO), in physics also known as Harper's model, where $v(x) = 2 \lambda \cos(2 \pi x)$ and $\lambda > 0$. 

An interesting phenomenon encountered for quasi-periodic Schr\"odinger operators are {\em{metal insulator transitions}}. Depending on the coupling constants, presence of the external magnetic field may enhance or deplete the conductivity in the crystal. The prototype is the AMO where the spectral properties pass from purely absolutely continuous (ac) spectrum for $\lambda < 1$ (``subcritical regime'') to pure point spectrum with exponentially localized eigenfunctions when $\lambda > 1$ (``supercritical regime''); the transition is marked by a ``critical point'' at $\lambda = 1$ where the spectrum is purely singular continuous (sc). For a review of the results known for the AMO we refer to \cite{Jitomirskaya_review_2007, JitomirskayaMarx_review_2015}.

A dynamical measure for such transitions is given by the {\em{Lyapunov exponent}} (LE), in physics interpreted as an inverse localization length, which quantifies the averaged asymptotics of the solutions to the time-independent Schr\"odinger equation. Whereas positivity of the LE is heuristically associated with localization, zero LE is interpreted to indicate delocalization. 

Solutions to the time-independent Schr\"odinger equation are obtained most conveniently in dynamical systems terms. Given an initial condition $(\psi_0,\psi_{-1})^\mathrm{T}$, $n-$step \emph{transfer matrices} $B_n^E(x;\alpha)$ allow to iteratively generate solutions of \\ $H_{v(x); \alpha} \psi = E \psi$ over $\mathbb{C}^\mathbb{Z}$ via
\begin{equation} \label{eq_solns}
\begin{pmatrix}
\psi_n\\\psi_{n-1}
\end{pmatrix}=B_n^E(x;\alpha)\begin{pmatrix}
\psi_0\\\psi_{-1}
\end{pmatrix} ~\mbox{, } B_n^E(x;\alpha):=\prod_{j=n-1}^{0}B^E(x+j\alpha) ~\mbox{,}
\end{equation}
where
\begin{equation}\label{eqn_cocycle}
B^E(x):=\begin{pmatrix}E-v(x)&-1\\1&0\end{pmatrix}.
\end{equation}
The {\em{Schr\"odinger cocycle}} $(\alpha,B^E)$, a dynamical system on $\mathbb{T} \times \mathbb{C}^2$ defined by $(x,v)\mapsto(x+\alpha,B^E(x)v)$, captures the iterative scheme in (\ref{eq_solns}) in a compact way. In particular, the Lyapunov exponent of a quasi-periodic Schr\"odinger operator is defined by
\begin{align} \label{eq_complexLE}
L(\alpha,B^E)&:= \lim_{n\to\infty}\frac{1}{n}\int_{\mathbb{T}}\log\|B_n^E(x;\alpha)\|\,dx ~\mbox{.}
\end{align}
Since $B^E \in SL(2, \mathbb{R})$, the LE of Schr\"odinger operators is always non-negative.

In his seminal work titled ``{\em{Global theory of one-frequency operators}}'' \cite{Avila_globalthy_published}, A. Avila introduces a framework that allows to appropriately generalize the metal-insulator transition observed for the AMO to arbitrary analytic potentials $v$. Relying on the analyticity of $v$, he considers the LE of the cocycle $(\alpha, B^E( . + i \epsilon))$ obtained by complexifying the phase in (\ref{eq_complexLE}); we will refer to this as the {\em{complexified LE}} and denote it by $L(\epsilon; E)$. 

Characterized by the behavior of the complexified LE about $\epsilon = 0$, Avila decomposes the spectrum $\Sigma$ into three mutually disjoint sets: {\em{supercritical, subcritical and critical energies}}. An energy is classified as supercritical if the complexified LE vanishes in a neighborhood of $\epsilon=0$, as supercritical if the LE is positive at $\epsilon=0$, and as critical if the LE is zero at $\epsilon=0$ but not subcritical. These three possible situations for fixed energy $E \in \Sigma$  are shown schematically in Fig. \ref{fig:complexified-LE-energies}. 

\begin{figure}[h!]
\begin{subfigure}[b]{.3\textwidth}
\begin{tikzpicture}
\draw [thick,->] (-2,0)--(2,0);
\draw [thick,->] (0,-.5)--(0,2);
\node [above] at (0,2) {\footnotesize$L(\varepsilon;E)$};
\node [below] at (2,0) {$\varepsilon$};
\draw [line width=1.2pt] (-2,1.5)--(-.5,0)--(.5,0)--(2,1.5);
\end{tikzpicture}
\caption{subcritical behavior}
\end{subfigure} 
\quad
\begin{subfigure}[b]{.3\textwidth}
\begin{tikzpicture}
\draw [thick,->] (-2,0)--(2,0);
\draw [thick,->] (0,-.5)--(0,2);
\node [above] at (0,2) {\footnotesize$L(\varepsilon;E)$};
\node [below] at (2,0) {$\varepsilon$};
\draw [line width=1.2pt] (-2,2)--(0,0)--(2,2);
\end{tikzpicture}
\caption{critical behavior}
\end{subfigure}
\quad
\begin{subfigure}[b]{.3\textwidth}
\begin{tikzpicture}
\draw [thick,->] (-2,0)--(2,0);
\draw [thick,->] (0,-.5)--(0,2);
\node [above] at (0,2) {\footnotesize$L(\varepsilon;E)$};
\node [below] at (2,0) {$\varepsilon$};
\draw [line width=1.2pt] (-1.5,2)--(0,.5)--(1.5,2);
\end{tikzpicture}
\caption{supercritical behavior}
\end{subfigure}
\caption{Behavior for energies in the spectrum, corresponding to local behavior (at $\varepsilon=0$) of the complexified LE.}
\label{fig:complexified-LE-energies}
\end{figure}
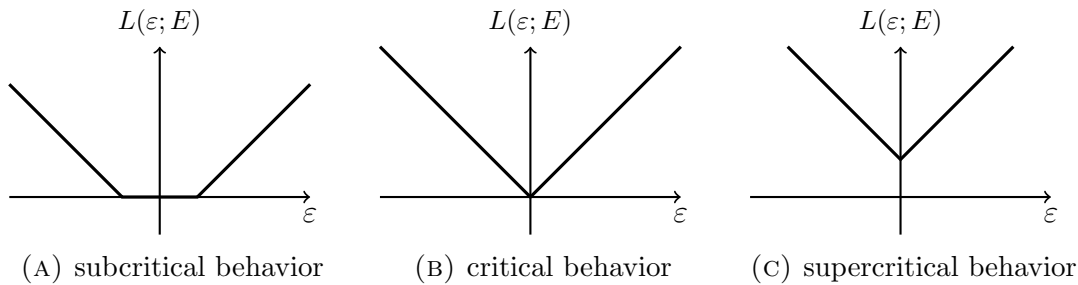

The supercritical regime just recovers the set of positive LE. It is the set of zero LE, however, for which Avila's decomposition yields additional insight unavailable prior to his global theory. Whereas from Kotani-Simon theory it has been known that the set of zero LE forms a Lebesgue essential support for the absolutely continuous (ac) spectrum \cite{Simon_1983, Kotani_1982}, it still leaves unaddressed sets of Lebesgue measure zero where the spectrum could potentially be singular. 

Further decomposing into subcritical and critical energies, enables to explicitly separate ac spectrum from singular spectrum. 
Here, as common, singular spectrum is defined as the union of singular continuous and pure point spectrum.
By the {\em{almost reducibility theorem}} \cite{Avila_prep_ARC_1,Avila_prep_ARC_2} the spectrum is purely ac on the set of subcritical energies. It is purely singular\footnote{Notice that we do not claim that the spectrum on the set of critical energies is purely singular continuous but only that it is purely singular (see the definition of singular spectrum, given above) for all phases. Even though the former is a known conjecture \cite{JitomirskayaMarx_review_2015}, so far this has not been proven yet. Excluding appearance of eigenvalues for all phases is a delicate and difficult problem, that has not even been established for the critical AMO.} on the set of critical energies, a consequence of the dynamical dichotomy of Avila-Fayad-Krikorian \cite{AvilaFayadKrikorian_2011}. Notably, both these spectral results for (\ref{eq_qso}) hold for {\em{all}} phases and {\em{all}} irrational frequencies. 

In this paper we focus on the set of zero LE and establish a sufficient criterion for subcritical behavior if the potential $v$ is a real trigonometric polynomial,
\begin{equation} \label{eq_potential}
v(x) = 2 \sum_{n=1}^M \left( a_n \cos(2 \pi n x) + b_n \sin(2 \pi n x) \right) ~\mbox{, } \vert a_M \vert + \vert b_M \vert > 0 ~\mbox{.}
\end{equation}
Here, we may assume absence of a constant term which would only result in a shift of the spectrum. 

We mention that detecting critical energies is in principle much more delicate, since in contrast to both sub- and supercriticality, criticality is not stable w.r.t. perturbations in $\alpha$ and $v$ \cite{Avila_globalthy_published}. In fact, Avila shows that small perturbations in the Fourier-coefficients of $v$ destroy critical behavior which allows to prove that for a measure theoretically typical (=prevalent) potential $v$ in the analytic category, the set of critical energies is empty \cite{Avila_globalthy_published}!

Our sufficient criterion relies on quantifying the asymptotics of the complexified LE, $L(\epsilon; E)$, as $\vert \epsilon \vert \to \infty$, building on earlier ideas (``{\em{method of almost constant cocycles}}'') which allowed to determine $L(\epsilon; E)$ for the AMO and extended Harper's model \cite{JitomirskayaMarx_2012}. This is achieved by imposing a suitable largeness condition on
\begin{equation} \label{eq_defnminfunc}
m(\epsilon; E):= \min_{x \in \mathbb{T}} \vert E - v(x + i \epsilon) \vert ~\mbox{,}
\end{equation}
associated with the upper left entry of the matrix (\ref{eqn_cocycle}). We mention that such largeness conditions have played a role earlier in proving positivity \cite{SoretsSpencer_1991, KleinDuarte_2013_posLE} as well continuity of the LE \cite{Bourgain_2005, AvilaJitomirskayaSadel_2013}.

Since $v$ is a trigonometric polynomial, $m(\epsilon; E) \to +\infty$ as $\vert \epsilon \vert \to \infty$. Thus letting $\epsilon_H= \epsilon_H(E;  \{ a_j ; b_j \}_{j=1}^{M})$ denote the {\em{largest}} $\epsilon \geq 0$ such that $m(\epsilon; E) = 2$, we will show that 
\begin{equation} \label{eq_complexherman}
L(\epsilon; E) = \log \vert i a_M - b_M \vert + 2 \pi M \vert \epsilon \vert ~\mbox{, for $\vert \epsilon \vert \geq \epsilon_H$. }
\end{equation}
As we will argue, (\ref{eq_complexherman}) in particular recovers the well-known lower bound for the LE due to Herman \cite{Herman_1983}
\begin{equation} \label{eq_herman}
L(\alpha, B^E) \geq \log \vert i a_M - b_M \vert ~\mbox{.}
\end{equation}
For this reason, we will refer to (\ref{eq_complexherman}) as {\em{complex Herman formula}} and to $\epsilon_H$ as the {\em{Herman radius}}. Note that as opposed to Herman's lower bound in (\ref{eq_herman}), the complex Herman formula does depend on both $E$ as well as on {\em{all}} the Fourier-coefficients through the Herman radius.

Exploring properties of the complexified LE for $E \in \Sigma$, will result in the following criterion for subcriticality, which constitutes our main result:
\begin{theorem} \label{thm_main}
Given a quasi-periodic Schr\"odinger operator, $\alpha$ irrational, and $v$ as in (\ref{eq_potential}).
\begin{itemize}
\item[(i)] $E \in \Sigma$ is subcritical if the Herman radius, $\epsilon_H = \epsilon_H(E;  \{ a_j ; b_j \}_{j=1}^{M})$, satisfies
\begin{equation} \label{eq_subcritcrit_energy}
\epsilon_H < -\dfrac{\log \vert i a_M - b_M \vert}{2 \pi (M - d)} ~\mbox{,}
\end{equation}
where 
\begin{equation} \label{eq_gcd}
d: = \mathrm{gcd}\{ n \in \{1, \dots, M\} : \vert a_n \vert + \vert b_n \vert > 0 \} ~\mbox{.}
\end{equation}

\item[(ii)] Define the {\em{uniform Herman radius}} $\epsilon_{H; {\mathrm{unif}} } = \epsilon_{H; {\mathrm{unif}} }(\{ a_j ; b_j \}_{j=1}^{M})$ to be the largest $\epsilon \geq 0$ such that
\begin{equation*} \label{eq_unifHerman}
\widetilde{m}(\epsilon; E) := \min_{ x \in \mathbb{T}} \vert v(\theta + i \epsilon) \vert = 4 + 2 \sum_{n=1}^{M} \left( \vert i a_n - b_n \vert \right) ~\mbox{.}
\end{equation*}
All energies in the spectrum are subcritical if
\begin{equation} \label{eq_subcritcrit_uniform}
\epsilon_{H; {\mathrm{unif}} } < -\dfrac{\log \vert i a_M - b_M \vert}{2 \pi (M - d)} ~\mbox{.}
\end{equation}
\end{itemize}
\end{theorem}
\begin{remark} \label{remark_mainthm}
\begin{itemize}
\item[(i)] The conditions (\ref{eq_subcritcrit_energy}) and (\ref{eq_subcritcrit_uniform}) implicitly assume that $\vert i a_M - b_M \vert \leq 1$ which is no restriction since the classical Herman bound (\ref{eq_herman}) implies supercritical behavior if $\vert i a_M - b_M \vert > 1$.
\item[(ii)] We note that proving that some $E_0$ is subcritical, implies existence of {\em{some}} ac spectrum about $E_0$ (see also Sec. \ref{sec_someapplications_energy}).
\end{itemize}
\end{remark}
We mention that both $\epsilon_H$ and $\epsilon_{H; {\mathrm{unif}}}$ can be estimated easily through polynomial root bounds; we discuss this in Sec. \ref{sec_someapplications}.

While the general theory is well developed, the framework of Avila's global theory has only been employed explicitly to the AMO and a generalization known as extended Harper's model \cite{JitomirskayaMarx_2012, AvilaJitomirskayaMarx_preprint_2015}. Physically interesting models include a special case of (\ref{eq_potential}) with all $b_n = 0$ (e.g. \cite{HiramotoKohmoto_1989, SoukoulisEconomou_1982, ChaoRiklundLiu_1985,ChaoRiklundLiu_1986}), also known as generalized Harper's model; here the few available rigorous results \cite{Herman_1983, HaroPuig_2013} focus on positivity of the LE. In light of proving subcritical behavior, we mention a related result on purely ac spectrum for potentials of the form $v(x) = \lambda f(x)$ where $f$ is a real analytic function and $\lambda \in \mathbb{R}$. Here, Bourgain-Jitomirskaya prove existence of $\lambda_0 = \lambda_0(f)$ such that for all $\vert \lambda \vert < \lambda_0$, the spectrum of (\ref{eq_qso}) is purely ac for {\em{a.e.}} $x \in \mathbb{T}$ if $\alpha$ is {\em{Diophantine}} \cite{BourgainJitomirskaya_2002_inventiones}. As this result is based on proving localization for the dual operator, it is bound to impose arithmetic conditions on both the frequency and the phases. Establishing subcritical behavior, however, implies results on ac spectrum irrespective of such arithmetic conditions.

We structure the paper as follows: Sec. \ref{sec_complexHerman} proves the main result, Thm. \ref{thm_main}, based on the complex Herman formula. As discussed there, the latter is an expression of asymptotic uniform hyperbolicity of the Schr\"odinger cocycle which is quantified by the Herman radius. The key ingredient here is Proposition \ref{prop_key} asserting that the Schr\"odinger cocycle is uniformly hyperbolic whenever $m(\epsilon; E) > 2$. 

Sec. \ref{sec_uniformds} contains a dynamical proof of Proposition \ref{prop_key} based on verifying a cone condition; the latter also allows to extract further estimates of the complexified LE (Proposition \ref{prop_key_amended}), thereby amending the result in Proposition \ref{prop_key}. The dynamical approach of Sec. \ref{sec_uniformds} is contrasted with a spectral theoretical proof of Proposition \ref{prop_key}, which in particular sheds a light on the spectral theoretic meaning of the lower bound ``2'' in the largeness condition on $m(\epsilon;E)$; as explained there, complexifying the phase leads to deformation of the spectrum of $H_{v(x); \alpha}$, thereby pushing a given energy $E$ into the resolvent set if $m(\epsilon; E) > 2$. 

In Sec. \ref{sec_someapplications} we present various applications of our main result, Thm. \ref{thm_main} to models of physical interest, among them to the generalized Harper model. Here, estimating the Herman radius is shown to be reduced to bounds on the largest positive root of a real polynomial, the latter of which are well explored in the literature.

It is natural try to extend Theorem \ref{thm_main} to Jacobi operators, which generalize quasi-periodic Schr\"odinger operators by introducing an additional trigonometric polynomial $c(x)$ whose evaluation modifies the discrete Laplacian, see (\ref{eq_jacobiop}). The extension is not immediate and is discussed in Sec. \ref{sec_jacobi}. It leads to distinguishing three cases depending on the relative degree of $c(x)$ and $v(x)$. 

We conclude the paper with some remarks on how one could use the ideas from Sec. \ref{sec_uniformds} to obtain conclusions about {\em{supercritical}} behavior (i.e. positivity of the LE) for quasi-periodic Schr\"odinger operators. Recently we learned that Jitomirskaya-Liu obtained the first quantitive results on positivity of the LE for the potential $v(x) = 2 (\lambda_1 \cos(2 \pi x) + \lambda_2 \cos(4 \pi x))$ which go beyond the classical Herman bound in (\ref{eq_herman}) \cite{JitomirskayaLiu_2015}. Even though our method yields a obtain a general lower bound for the LE improving Herman (see Proposition \ref{prop_positiveLE}), it unfortunately proved to be difficult to extract quantitative results for a concrete potential. On the other hand, the lower bound in Proposition \ref{prop_positiveLE} can be easily analyzed numerically, thereby giving rise at least to a simple numerical scheme to test for super-criticality.

\subsection{Some notation}
As common, for $\delta >0$ and $X = \mathbb{R}, \mathbb{C}, M_2(\mathbb{C})$, the space $\mathcal{C}_\delta^\omega(\mathbb{T}, X)$ denotes the $X$-valued functions on $\mathbb{T}$ with holomorphic extension to a neighborhood of $\abs{\im{z}} \leq \delta$, $\delta>0$ equipped with supremum norm. To obtain statements independent of $\delta$, we consider $\mathcal{C}^\omega(\mathbb{T};X):=\cup_{\delta>0} \mathcal{C}_\delta^\omega(\mathbb{T},X)$ with the {\em{inductive limit topology}} induced by $\Vert . \Vert_\delta$. In this topology, convergence of a sequence $f_n \to f$ is equivalent to existence of some $\delta>0$ such that $f_n \in \mathcal{C}_\delta^\omega(\mathbb{T}, X)$ eventually and $\Vert f_n - f \Vert_\delta \to 0$ as $n \to \infty$.

\section{The complex Herman formula} \label{sec_complexHerman}

Fix $E \in \mathbb{R}$. The complex Herman formula (\ref{eq_complexherman}) rests on the basic observation that if $v$ is given by (\ref{eq_potential}), the upper left corner in (\ref{eqn_cocycle}) will dominate the Schr\"odinger cocycle as $\vert \epsilon \vert \to \infty$. Specifically, complexifying the phase and taking out the dominating term yields
\begin{equation} \label{eq_almostconst}
B^E(x + i \epsilon) =  \left( a_M - \frac{b_M}{i} \right) \mathrm{e}^{2 \pi M \epsilon} \mathrm{e}^{- 2 \pi i M x} \begin{pmatrix} -1 + o(1) & o(1) \\ o(1) & 0 \end{pmatrix} ~\mbox{,}
\end{equation}
uniformly in $x \in \mathbb{T}$, as $\epsilon \to +\infty$. 

Thus the cocycle dynamics is asymptotically determined by the {{\em{almost constant cocycle}}, $(\alpha, ( \begin{smallmatrix} - 1 + o(1) & o(1) \\ o(1) & 0 \end{smallmatrix} ))$. Since $L( \alpha, ( \begin{smallmatrix} -1 & 0 \\ 0 & 0 \end{smallmatrix} ) ) = 0$, continuity of the LE in the analytic category \cite{JitomirskayaMarx_2012,AvilaJitomirskayaSadel_2013} implies
\begin{equation}
L(\epsilon; E) = \log \vert i a_M - b_M \vert + 2 \pi M \vert \epsilon \vert + o(1) ~\mbox{, as $\epsilon \to \infty$ .}
\end{equation}

On other hand, from Avila's global theory the analytic properties of the complexfied LE are well understood:
\begin{theorem}[\cite{Avila_globalthy_published}] \label{thm_complexLEglobal}
Let $\alpha$ irrational and $E \in \mathbb{R}$. Then, $\epsilon \mapsto L(\epsilon; E)$ is convex, even, non-negative, and piecewise linear with right-derivatives satisfying
\begin{equation} \label{eq_accel}
\omega(\epsilon; E) := \frac{1}{2 \pi} D_+ L(\epsilon; E) \in \mathbb{Z} ~\mbox{.}
\end{equation}
\end{theorem}
\begin{remark} \label{rem_quantiz}
The quantity defined in (\ref{eq_accel}) is known as the {\em{acceleration}}. Following the proof of (\ref{eq_accel}) given in \cite{Avila_globalthy_published} actually shows that if $v$ is $1/d$-periodic for some $d \in \mathbb{N}$, then 
$\omega(\epsilon; E) \in d \mathbb{Z}$. In particular for $v$ of the form (\ref{eq_potential}), the least {\em{positive}} value of the acceleration is $d$ defined in (\ref{eq_gcd}). As we will see, this accounts for the appearance of $d$ in Theorem \ref{thm_main}.
\end{remark}

From Theorem \ref{thm_complexLEglobal} we thus conclude existence of some $\epsilon_0 \geq 0$ such that 
\begin{equation} \label{eq_complexherman_1}
L(\epsilon; E) = \log \vert i a_M - b_M \vert + 2 \pi M \vert \epsilon \vert ~\mbox{, for all $\vert \epsilon \vert \geq \epsilon_0$.}
\end{equation}
Note that by convexity, the asymptotic formula (\ref{eq_complexherman_1}) automatically implies a global lower bound, which, letting $\epsilon = 0$, recovers the original Herman bound in (\ref{eq_herman}).

We mention that above argument was first used in \cite{JitomirskayaMarx_2012} to study extended Harper's model,  a Jacobi operator generalizing the AMO. There, as a result of (\ref{eq_accel}), the limited values of the acceleration allowed to extrapolate the asymptotics to obtain an expression for $L(\epsilon;E)$ valid for {\em{all}} $\epsilon \in \mathbb{R}$. Using Remark \ref{rem_quantiz}, this analysis has an immediate extension to quasi-periodic Schr\"odinger operators if $v$ in (\ref{eq_potential}) has only {\em{one}} non-vanishing term, in which case for all $E \in \Sigma$,
\begin{equation}
L(\epsilon; E) = \max \{ \log \vert i a_M - b_M \vert + 2 \pi M \vert \epsilon \vert ; 0\} ~\mbox{, all $\epsilon \in \mathbb{R}$.}
\end{equation} 
In particular, the situation is analogous to the AMO, i.e. all energies are subcritical for $\vert i a_M - b_M \vert < 1$, critical if $\vert i a_M - b_M \vert =1$, and supercritical if $\vert i a_M - b_M \vert > 1$. 

For more general $v$, the simple idea underlying Theorem \ref{thm_main} is to gain additional information about the complexified LE by {\em{quantifying}} when the asymptotic formula in (\ref{eq_complexherman_1}) holds or, put equivalently, when $\epsilon \mapsto L(\epsilon; E)$ is eventually linear.

To this end, we take advantage of a key result in \cite{Avila_globalthy_published} which characterizes the linear and positive segments of $\epsilon \mapsto L(\epsilon; E)$ by uniform hyperbolicity of the Schr\"odinger cocycle. The following provides a sufficient criterion for uniform hyperbolicity, therefore helps to identify linear pieces in the complexified LE. We recall the definition of the auxiliary function $m(\epsilon;  E)$ in (\ref{eq_defnminfunc}).
\begin{prop} \label{prop_key}
Let $v \in \mathcal{C}^\omega(\mathbb{T}; \mathbb{R})$, $\alpha \in \mathbb{T}$ irrational, and $E \in \mathbb{R}$. The Schr\"odinger cocycle is uniformly hyperbolic whenever $m(\epsilon; E) > 2$.
\end{prop}
\begin{remark} \label{rem_prop_key_optimality}
The lower bound ``2'' of $m(\epsilon; E)$ is {\em{optimal}} in general. For instance if $v \equiv 0$ and $E = \pm 2$, $(\alpha, B^E)$ cannot be uniformly hyperbolic since $\Sigma = [-2,2]$ and uniform hyperbolicity is known to be an open property which, by Johnson's theorem \cite{Johnson_1986}, cannot occur on the spectrum. The optimality will also follow directly from the proof given in Sec. \ref{sec_uniformds}, see Remark \ref{rem_lem_stability_optimality}.
\end{remark}
We postpone the proof of Proposition \ref{prop_key} for now (see Sec. \ref{sec_uniformds} and \ref{sec_propkeySpec}) and rather turn to showing how it implies our main result, Theorem \ref{thm_main}. 

\begin{proof}[Proof of Theorem \ref{thm_main}]
First observe that for $\epsilon$ outside the radius of zeros of \\ $(E - v)$, the minimum modulus principle implies
\begin{equation}
m(\epsilon; E) = \min_{\vert \im(z) \vert  \geq \epsilon} \vert E - v(z) \vert ~\mbox{,}
\end{equation}
whence $m(\epsilon; E)$ increases {\em{strictly}} for $\epsilon$ outside the radius of zeros of $(E - v)$. In particular, the properties of $\epsilon \mapsto L(\epsilon;E)$ stated in Proposition \ref{prop_key} imply that (\ref{eq_complexherman_1}) holds with $\epsilon_0$ replaced by the Herman radius,  $\epsilon_H = \epsilon_H(E;  \{ a_j ; b_j \}_{j=1}^{M})$, introduced in Sec. \ref{sec_intro}.

To prove Theorem \ref{thm_main} (i), considering the contrapositive, if $E\in\Sigma(\alpha)$ is \emph{not} subcritical, Remark \ref{rem_quantiz} yields $\omega(\epsilon=0; E)\ge d$. Then, the complex Herman formula and convexity of $L(\epsilon;E)$ yield the {\em{upper}} bound,
\begin{align*}
0\le L(\epsilon;E)&\le\log| i a_M - b_M |+2\pi M \epsilon_H+2\pi d(\epsilon-\epsilon_H) ~\mbox{, $0 \leq \epsilon \leq \epsilon_H$ .}
\end{align*}

In particular,
\[
\log| i a_M - b_M |+2\pi\epsilon_H(M-d)\ge0,
\]
which is equivalent to
\[
\epsilon_H\ge-\frac{\log| i a_M - b_M |}{2\pi(M-d)} ~\mbox{.}
\]
Thus if $\epsilon_H<-\dfrac{\log| i a_M - b_M |}{2\pi(M-d)}$, $E$ must be subcritical.

Finally, we obtain the uniform criterion for subcritical behavior in Theorem \ref{thm_main} (ii), estimating the spectral radius of $H_{v(x); \alpha}$ from above by $2 + 2 \sum_{n=1}^{M} \left( \vert i a_n  - b_n \vert \right)$; this allows to eliminate the energy dependence by replacing $\epsilon_H$ with the uniform Herman radius ($\epsilon_{H; {\mathrm{unif}}}$) defined in (\ref{eq_unifHerman}).
\end{proof}

Theorem \ref{thm_main} has thus been reduced to proving Proposition \ref{prop_key}. We shall give two proofs, one dynamical and one spectral theoretical, the subjects of the Sec. \ref{sec_uniformds} and \ref{sec_propkeySpec}, respectively.

\section{Asymptotic domination} \label{sec_uniformds}
 
We recall that factoring the Schr\"odinger cocycle according to 
\begin{equation} \label{eq_cocyclefactorization}
B^E(x + i \epsilon) = \left( E - v(x + i \epsilon) \right) \begin{pmatrix} 1 & - \frac{1}{E - v(x + i \epsilon)} \\ \frac{1}{E - v(x + i \epsilon)} & 0 \end{pmatrix} =: \left( E - v(x + i \epsilon) \right) D_\epsilon(x) ~\mbox{,}
\end{equation}
shows that its asymptotic dynamics is determined by $(\alpha, D_\epsilon)$, which as $\epsilon \to +\infty$, is uniformly close to the constant cocycle $(\alpha,  (\begin{smallmatrix} 1 & 0 \\ 0 & 0 \end{smallmatrix} ))$. Trivially, the latter induces the invariant splitting $\C^2=\langle\binom{1}{0}\rangle\oplus\langle\binom{0}{1}\rangle$ such that the dynamics in one invariant subspace dominates the other. In this section, we prove Proposition \ref{prop_key} by showing that these dynamical features are in fact already present once $m(\epsilon; E) > 2$.

To capture these dynamical features precisely, we recall the following terminology from partially hyperbolic dynamics. Given $D \in \mathcal{C}^\omega(\mathbb{T}; M_2(\mathbb{C}))$ and $\alpha \in \mathbb{T}$, a cocycle $(\alpha,D)$ is said to induce a {\em{dominated splitting}} (also ``uniform domination;'' write $(\alpha,D)\in\mathcal{DS}$) if there exists a continuous, nontrivial splitting $\C^2=S_x^{(1)}\oplus S_x^{(2)}$ and $N\in\N$ satisfying
\begin{enumerate}
\item[(i)] $(\alpha,D)$-invariance, i.e. $D_N(x; \alpha) E_x^{(j)}\subseteq E_{x+N\alpha}^{(j)}$, for $1 \leq j \leq 2$,
\item[(ii)] for all $v_j\in S_x^{(j)}\setminus\{0\}$, $1 \leq j \leq 2$, one has
\begin{equation}
\frac{\Vert D_N(x; \alpha) )v_1\Vert}{\Vert v_1\Vert}>\frac{\Vert D_N(x; \alpha)v_2\Vert}{\Vert v_2\Vert} ~\mbox{.} \label{def:dominance}
\end{equation}
\end{enumerate}
Here, as in (\ref{eq_solns}), we denote $D_N(x ; \alpha): = \prod_{j = N-1}^{0} D(x + j \alpha)$. For obvious reasons, will refer to $S_x^{(1)}$ as the {\em{dominating section}} and to $S_x^{(2)}$ as the {\em{minoring section}}.

Clearly, the condition $(\alpha,D)\in\mathcal{DS}$ is equivalent to some iterate of $(\alpha, D)$ being continuously conjugate to a diagonal cocycle where one diagonal entry uniformly dominates the other, i.e. $\exists N \in\mathbb{N}$ and $C\in\mathcal{C}(\mathbb{T},GL(2,\mathbb{C}))$ such that
\begin{equation}
C(x+N\alpha)^{-1}D_N(x)C(x)=\begin{pmatrix}
\lambda_1(x)&0\\0&\lambda_2(x)
\end{pmatrix},\label{eqn:conjugate}
\end{equation}
with $\lambda_1,\lambda_2\in\mathcal{C}(\mathbb{T},\mathbb{C})$ such that for all $x\in\mathbb{T}$,
\begin{equation} \label{eq_ds}
\vert \lambda_1(x) \vert > \vert \lambda_2(x) \vert ~\mbox{.}
\end{equation}
We mention that for analytic cocycles it is well known that analyticity is inherited by the invariant splitting, which in turn gives rise to analyticity of the conjugacy, see e.g. \cite[Theorem 6.1]{AvilaJitomirskayaSadel_2013}. 

Obviously, $(\alpha,  (\begin{smallmatrix} 1 & 0 \\ 0 & 0 \end{smallmatrix} )) \in \mathcal{DS}$. For a cocycle $(\alpha, D)$ with $\vert \det D \vert \equiv 1$, $\mathcal{DS}$ reduces to the notion of {\em{uniform hyperbolicity}} ($\mathcal{UH}$), in which case (\ref{eq_ds}) simplifies to
\begin{equation} \label{eq_condunifh}
\inf_{x \in \mathbb{T}} \vert \lambda_1(x) \vert > 1 ~\mbox{.}
\end{equation}
Since $\mathcal{DS}$ is the appropriate notion for the non-invertible cocycle $(\alpha,  (\begin{smallmatrix} 1 & 0 \\ 0 & 0 \end{smallmatrix} ))$, it will however be more convenient in this section to work with the latter. From the factorization in (\ref{eq_cocyclefactorization}) it is clear that 
\begin{equation}
(\alpha, B_\epsilon^E) \in \mathcal{UH} \Leftrightarrow (\alpha, D_\epsilon) \in \mathcal{DS} ~\mbox{,}
\end{equation}
whence the proof of Proposition \ref{prop_key} is reduced to showing that 
\begin{equation}
(\alpha, D_\epsilon) \in \mathcal{DS} ~\mbox{whenever } m(\epsilon; E) > 2 ~\mbox{.}
\end{equation}

It is well known that $\mathcal{DS}$ is an open property in $\mathbb{T} \times \mathcal{C}^\omega(\mathbb{T}; M_2(\mathbb{C}))$ \cite{Ruelle_1979_LEanalytic, BonattiDiazViana_book_2005}, in particular, $(\alpha, D_\epsilon) \in \mathcal{DS}$ once $m(\epsilon; E)$ is sufficiently large. The point here is to {\em{quantify}} the neighborhood of stability for $\mathcal{DS}$ about $(\alpha,  (\begin{smallmatrix} 1 & 0 \\ 0 & 0 \end{smallmatrix} ))$, which will result in Proposition \ref{prop_key}. This will be done by verifying a {\em{cone condition}}, a well known strategy to detect presence of a dominated splitting. 

It will be useful to work in the projective plane $\mathbb{PC}^2$ which we identify with $\overline{\mathbb{C}}:= \mathbb{C} \cup \{ \infty \}$ via $(v_1, v_2) \mapsto \frac{v_2}{v_1}$ so that $D = \left( \begin{smallmatrix} a & b \\ c & d  \end{smallmatrix} \right) \in M_2(\mathbb{C})$ acts on $\overline{\C} \setminus \ker D$ as the fractional linear transformation
\begin{equation} \label{eq_matrixaction}
D \cdot z := \dfrac{c + dz}{a + bz} ~\mbox{.} 
\end{equation} 

Given a cocycle $(\alpha,D)$, a conefield for $(\alpha,D)$ is an open subset $U \subset \mathbb{T} \times \mathbb{PC}^2$ of the form $\cup_{x \in \mathbb{T}} \{x\} \times U_x$ such that, for all $x \in \mathbb{T}$, $\overline{U_x}$ is non-empty, properly contained in $\mathbb{PC}^2$, and $\overline{U_x} \cap \ker D(x) = \emptyset$. A conefield $U = \cup_{x \in \mathbb{T}} \{x\} \times U_x$ for $(\alpha,D)$ is said to satisfy a cone condition if there exists $N \in \mathbb{N}$ such that for every $x \in \mathbb{T}$, one has that $D_N(\alpha; x) \cdot \overline{U_x} \subseteq U_{x + N\alpha}$. It is known (see e.g. \cite{Avila_2011, AvilaJitomirskayaSadel_2013}) that verifying a cone condition implies $\mathcal{DS}$.

Using (\ref{eq_matrixaction}), the proof of Proposition \ref{prop_key} is hence reduced to the following simple contraction Lemma:

\begin{lemma} \label{lem_stability}
For $\epsilon > 0$ and $0 \leq \delta < 1$, consider the class of matrices
\begin{equation}
\mathcal{S}_{\epsilon, \delta} : =  \left\{ \left( \begin{smallmatrix} 1 & b \\ c & d \end{smallmatrix} \right) \in M_2(\mathbb{C}) : \, |d| \leq \delta; ~|b| < \epsilon; \, |c| \epsilon < \dfrac{(1-\delta)^2}{4}\right\} ~\mbox{.}
\end{equation}
Then for each $D \in \mathcal{S}_{\epsilon, \delta}$, the map $F_D: \overline{B_r(0)} \to \overline{B_r(0)}$, $F_D(z):= D \cdot z$ is a contraction, where $r = r(\epsilon, \delta) := \frac{1-\delta}{2 \epsilon}$. 
\end{lemma} 
\begin{remark} \label{rem_lem_stability_optimality}
The conditions on $\vert b \vert$, $\vert c \vert$ in Lemma \ref{lem_stability} are in general optimal. For instance, direct computation shows that if $D = (\begin{smallmatrix} 1 & -1/2 \\ 1/2 & 0 \end{smallmatrix})$, $F_D$ is not a contraction for any $r > 0$, cf. also Remark \ref{rem_prop_key_optimality}.
\end{remark}

\begin{proof}
Let $D \in \mathcal{S}_{\epsilon, \delta}$. Write $r_{\gamma} := \frac{\gamma}{\epsilon}$, for some $\gamma \in (0,1)$ to be determined later. If $|z| \leq r_{\gamma}$, then 
\begin{equation}
|F_D(z)| = \dfrac{|c+dz|}{|1+bz|} \leq \dfrac{1}{1-\gamma}\left(|c| + \frac{\delta \gamma}{\epsilon}\right) ~\mbox{.}
\end{equation}

Thus, $F_D$ maps $\overline{B_{r_{\gamma}}(0)}$ to itself, if the parameters satisfy
\begin{equation} \label{eq_cond1}
\epsilon |c| \leq \gamma ((1-\gamma)-\delta) ~\mbox{.} 
\end{equation}

For fixed $\delta \in [0,1)$, the right hand side of (\ref{eq_cond1}) is maximized when $\gamma = \dfrac{1-\delta}{2}=: \gamma^*$, so that the condition in (\ref{eq_cond1}) becomes
\begin{equation}
\epsilon \vert c \vert \leq \frac{(1 -\delta)^2}{4} ~\mbox{.}
\end{equation}

On the other hand, one has 
\begin{equation} 
\sup_{z, w \in \overline{B_{r_{\gamma}}(0)}} \dfrac{|F_D(z) - F_D(w)|}{|z - w|} = \sup_{z \in \overline{B_{r_{\gamma}}(0)}}|F'_D(z)| = \sup_{z \in \overline{B_{r_{\gamma}}(0)}} \dfrac{|d-bc|}{|1+bz|^2} \leq \dfrac{\delta +\epsilon |c|}{(1-\gamma)^2} ~\mbox{.}
\end{equation}
Thus, for $F_D$ to be a contraction on $\overline{B_{r_{\gamma}}(0)}$ it suffices to have
\begin{equation}
\delta + \epsilon |c| < (1- \gamma)^2 ~\mbox{,} 
\end{equation}
which, taking $\gamma = \gamma^*$, becomes
\begin{equation}
\epsilon |c| < \dfrac{(1-\delta)^2}{4} ~\mbox{,}
\end{equation} 
in agreement with the definition of $\mathcal{S}_{\epsilon, \delta}$. 
\end{proof}

Finally, we mention that Lemma \ref{lem_stability} allows to extract an estimate for the complexified Lyapunov exponent. To this end, we note that the proof of Lemma \ref{lem_stability} shows that if $D = (\begin{smallmatrix} 1 & b \\ c & d \end{smallmatrix})$ satisfies $\vert b \vert , \vert c \vert \leq \eta < \frac{1 - \delta}{2}$, then $F_D$ is a contraction on $\overline{B_1(0)}$ with contraction constant
\begin{equation} \label{eq_contractionconst}
K \leq \frac{\delta + \eta^2}{(1 - \eta)^2} ~\mbox{.}
\end{equation}


\begin{prop}[Proposition \ref{prop_key} amended] \label{prop_key_amended}
Let $v \in \mathcal{C}^\omega(\mathbb{T}; \mathbb{R})$, $\alpha \in \mathbb{T}$ irrational, and $E \in \mathbb{R}$. The Schr\"odinger cocycle is uniformly hyperbolic whenever \newline $m(\epsilon; E) > 2$, in which case 
\begin{equation} \label{eq_complexLE_ds_bd}
L(\epsilon; E) = \int_\mathbb{T} \log \vert E - v(x + i \epsilon) \vert ~\ud x + \Xi ~\mbox{,} 
\end{equation}
where
\begin{eqnarray} \label{eq_complexLE_ds_bd_1}
~\quad & ~~ \log \left\{ \dfrac{ (1 - \frac{\sigma(\epsilon;E)}{m(\epsilon;E)})^2 + \frac{1}{m(\epsilon;E)^2} }{1 + \sigma(\epsilon;E)^2}  \right\} \leq 2 \Xi \leq \log \left\{ \left(1 + \frac{\sigma(\epsilon;E)}{m(\epsilon;E)}\right)^2 + \frac{1}{m(\epsilon;E)^2} \right\}     ~\mbox{, } \\ 
& \sigma(\epsilon; E) = \min\left\{ 1 ~;~ \dfrac{m(\epsilon;E)-1}{m(\epsilon;E)(m(\epsilon;E)-2)} \right\} \label{eq_estimatesolutionFP} ~\mbox{.}
\end{eqnarray}
\end{prop}
\begin{remark}
It is clear from the arguments above that 
\begin{equation}
\Xi = L(\alpha, D_\epsilon) = 0 ~\mbox{, for $\vert \epsilon \vert \geq \epsilon_H$ .}
\end{equation}
Proposition \ref{prop_key_amended} provides additional information {\em{outside}} the asymptotic regime, i.e. if $m(\epsilon_1;E) > 2$ for some $0 \leq \vert \epsilon_1 \vert < \epsilon_H$. In particular, for such $\epsilon_1$, (\ref{eq_complexLE_ds_bd_1}) implies $\Xi \geq - \log(2)$, or
\begin{equation} \label{eq_complexLE_ds_lowerboundsupercrit}
L(\epsilon_1; E) \geq  \int_\mathbb{T} \log \vert E - v(x + i \epsilon_1) \vert ~\ud x - \log(2) ~\mbox{.}
\end{equation}
We will return to this in Sec. \ref{sec_supercrit}.
\end{remark}

\begin{proof}
Fix $\epsilon > 0$ such that $m(\epsilon; E) > 2$. First observe that if $w_\epsilon: \mathbb{T} \to \mathbb{C}^2$ is any continuous lift of the dominating section $S_1(x; \epsilon)$ with $\|w_\epsilon(x)\| \equiv 1$, the complexified LE is given by
\begin{equation}
L(\epsilon; E)=\int_{\mathbb{T}}\log\|B^E_\epsilon(x) w_\epsilon(x) \| ~\ud x ~\mbox{.}
\end{equation}
In particular, normalizing $\left(\begin{smallmatrix}1\\S_1(\epsilon; x)\end{smallmatrix}\right)\in\mathbb{C}^2$, the factorization in (\ref{eq_cocyclefactorization}) yields (\ref{eq_complexLE_ds_bd}) with $\Xi = L(\alpha, D_\epsilon)$ given by
\begin{equation}\label{eqn:LE-logE}
\Xi = \int_{\mathbb{T}} \log \Vert D_\epsilon(x) \begin{pmatrix} 1 \\ S_1(\epsilon; x) \end{pmatrix} \Vert \ud x - \frac{1}{2} \int_\mathbb{T} \log ( 1 + \vert S_1(\epsilon; x) \vert^2 ) \ud x ~\mbox{.}
\end{equation}

To estimate $\Xi$, note that from the cone condition, $S_1(x; \epsilon)$ is determined by the fixed point problem $D_\epsilon(x) \cdot S_1(x;\epsilon) = S_1(x+\alpha; \epsilon)$, which, by Lemma \ref{lem_stability} has a solution in $\mathcal{C}^\omega(\mathbb{T}; \overline{B_1(0)})$ since
\begin{equation} \label{eq_prop_key_amended_1}
d = 0 ~\mbox{, } \vert b \vert = \vert c \vert = \frac{1}{\vert E- v(x + i \epsilon)\vert} \leq \frac{1}{m(\epsilon;E)} < \frac{1}{2} ~\mbox{.}
\end{equation}

We use the following standard fact from Banach fixed point theory to obtain an upper bound for $\vert S_1(x; \epsilon) \vert$:
\begin{fact}\label{fact:contraction}
Let $T_1,T_2:X\to X$ be two contractions with contraction constants $K_1,K_2$ on a complete metric space $(X,\rho)$. Denote by $x_j^*$ the (unique) fixed point of $T_j,1\leq j \leq 2$. If $\gamma:=\sup_{x\in X}\rho(T_1x,T_2x)<\infty$, then
\begin{equation}
\rho(x_1^*,x_2^*)\leq \gamma \min\big\{\frac{1}{1-K_1},\frac{1}{1-K_2}\big\}.
\end{equation}
\end{fact}
Thus, comparing the solutions $S_1(x; \epsilon)$ and $S_1(x, \epsilon = +\infty) = 0$, Fact \ref{fact:contraction} and (\ref{eq_contractionconst}) with
\begin{equation}
\gamma = \sup_{x \in \mathbb{T}, z \in \overline{B_1(0)}} \vert F_{D_{\epsilon}(x)} (z)- F_{D_{\epsilon = \infty}(x)}(z) \vert \leq \frac{1}{m(\epsilon;E) - 1} ~\mbox{,}
\end{equation}
implies $\vert S_1(x; \epsilon) \vert \leq \sigma(\epsilon; E)$ where $\sigma(\epsilon; E)$ is given by (\ref{eq_estimatesolutionFP}). In particular, the upper bound on $\Xi$ in (\ref{eq_complexLE_ds_bd_1}) follows immediately from (\ref{eqn:LE-logE}).

On the other hand, using (\ref{eq_prop_key_amended_1}), we estimate
\begin{equation}
\Xi \geq \frac{1}{2}\int_\mathbb{T} \log\frac{(1-|b| |S_1(\epsilon; x)|)^2+|b|^2}{1+|S_1(\epsilon; x)|^2}\ud x ~\mbox{,}
\end{equation}
which implies the lower bound in (\ref{eq_complexLE_ds_bd_1}) since on $[0, \frac{1}{m(\epsilon;E)}] \times [0, \sigma(\epsilon; E)]$, $(b,S) \mapsto \frac{(1-bS)^2+b^2}{1+S^2}$ is minimized at $b=\frac{1}{m(\epsilon;E)}$ and $s=\sigma(\epsilon; E)$.
\end{proof}

\section{Proof of Proposition \ref{prop_key} - spectral theory approach} \label{sec_propkeySpec}

In this section we present an alternative, spectral theoretic proof of Proposition \ref{prop_key}. Rather than verifying a cone condition as in Sec. \ref{sec_uniformds}, we will use Weyl $m$-functions to obtain explicit expressions for the invariant splitting giving rise to uniform hyperbolicity of the Schr\"odinger cocycle. In particular, this argument will shed a light on the spectral theoretic meaning of the lower bound ``2'' in the largeness condition on $m(\epsilon;E)$. We mention that many ideas in this section were inspired by our earlier work in \cite{Marx_2014}.

We start by noting that complexifying the phase in (\ref{eq_qso}) yields a discrete Schr\"{o}dinger operator with {\em{complex}} potential,
\begin{equation}\label{eqn:schrodinger2}
(H_{x+i\epsilon}\psi)(n):=((\Delta+V_{x+i\epsilon})\psi)(n) = \psi(n-1)+\psi(n+1)+v(x+n\alpha+i\epsilon)\psi(n) ~\mbox{,}
\end{equation}
in particular for $\epsilon \neq 0$, (\ref{eqn:schrodinger2}) is a {\em{non self-adjoint}} operator on $\mathit{l}^2(\mathbb{Z})$. Since both $\alpha$ and $v$ are considered to be fixed, we will simplify notation by dropping the explicit dependence on the frequency and the potential. Henceforth, we write $z:=x+i\epsilon$ and use $z$ and $x+i\epsilon$ interchangeably. 

Denote by $\delta_n$ the elements of the standard basis in $\mathit{l}^2(\mathbb{Z})$, and let $P^\pm$ be the orthogonal projection onto the subspaces, $\operatorname{Span}\{\delta_n: n > 0\}$ and $\operatorname{Span}\{\delta_n: n<0\}$, respectively. Define the half-line operators, $H_{z,\pm}:=P^{(\pm)}H_zP^{(\pm)}$.

For $E$ in the resolvent sets $\rho(H_{z, \pm})$, we let
\begin{align} \label{eq_definesection}
s_-(z, E)&:=-m_-(z,E) ~\mbox{,} \\
s_+(z, E)&:=-m_+(z - \alpha ,E)^{-1} ~\mbox{,}
\end{align}
where $m_\pm(z, E):=\langle\delta_{\pm1},(H_{z,\pm}-E)^{-1}\delta_{\pm1}\rangle$ are the {\em{Weyl $m-$functions}}. The resolvent identities show that $m(. + i \epsilon, E)$ and hence $s_\pm(. + i \epsilon, E)$ are continuous on $\mathbb{T}$. The main result in this section is the following angle formula:
\begin{lemma}[Angle formula]\label{lemma:angleformula}
Let $E \in \mathbb{R}$. If $\epsilon = \im z$ is such that $m(\epsilon; E) > 2$, then $E \in \rho(H_{z,+}) \cap \rho(H_{z,-}) \cap \rho(H_z)$ and 
\begin{equation}\label{eqn:goal}
|s_-(z, E)-s_+(z, E)|=|\langle\delta_0,(H_{z}-E)^{-1}\delta_0\rangle|^{-1} > 0 ~\mbox{.}
\end{equation}
Here, $\rho(H_z)$ denotes the resolvent set of the operator $H_z$.
\end{lemma}
Lemma \ref{lemma:angleformula} shows that the angle between the invariant sections $s_\pm(. + i \epsilon, E)$ is (uniformly) bounded away from zero, in particular $s_\pm(. + i \epsilon, E)$ give rise to a continuous, $(\alpha, B_\epsilon^E)$-invariant splitting of $\mathbb{C}^2$. 

We mention that (\ref{eqn:goal}) appeared earlier in \cite{Marx_2014} for complex {\em{energies}}\footnote{In fact, a continuity argument in \cite{Marx_2014} shows that for real phases, the angle formula extends to all $E$ is in the resolvent of the full line operator $H_x$.}  as apposed to complex phases. Below-mentioned argument will show that the underlying feature necessary in both cases is really that $E \in \rho(H_z)$, which of course is trivial for {\em{complex}} $E$ and {\em{real}} phase. For {\em{real}} $E$, {\em{complexifying}} the phase leads to deformation of the spectrum of $H_x$, pushing a given energy $E$ into the resolvent set if $m(\epsilon; E)$ is sufficiently large. 

\begin{proof}[Proof of the angle formula]
First, we verify $m(\epsilon; E) > 2$ implies $E\in\rho(H_{x+i\epsilon})$. The idea is to simply write
\begin{align} \label{eq_operatorfactorization}
(H_{x+i\epsilon}-E)=\Delta+V_{x+i\epsilon}-E&=(V_{x+i\epsilon}-E)[I+(V_{x+i\epsilon}-E)^{-1}\Delta] ~\mbox{,}
\end{align}
noticing that $m(\epsilon; E) = \min_{x\in\T}|v(x+i\epsilon)-E|>2$ guarantees existence of $(V_{z}-E)^{-1}$. (\ref{eq_operatorfactorization}) is really the operator analogue of the factorization in (\ref{eq_cocyclefactorization}). The operator $I+(V_{x+i\epsilon}-E)^{-1}\Delta$ is invertible if $\|(V_{x+i\epsilon}-E)^{-1}\Delta\|<1$, which is satisfied since $\|(V_{x+i\epsilon}-E)^{-1}\|<\frac{1}{2}$ and $\|\Delta\|=2$. Clearly, the same argument works for the half-line operators, showing that $m(\epsilon;E) > 2$ also implies $E\in\rho(H_{x+i\epsilon,\pm})$. In summary, all quantities in (\ref{eqn:goal}) are thus well-defined.

We next use some standard facts from the spectral theory of second order finite difference operators, usually formulated for the self-adjoint (Jacobi) case (see e.g. \cite{Teschl_book_2000}). Under the circumstances discussed here, everything is easily seen to carry over even though the operator (\ref{eqn:schrodinger2}) is not self-adjoint; for the reader's convenience, we summarize the necessary facts including brief arguments in the following paragraph. 

Denote the matrix elements of the Green's function by
\begin{equation}
G_{z}(E,n,m):=\langle\delta_n,(H_{z}-E)^{-1}\delta_m\rangle ~\mbox{, } n, m \in \mathbb{Z} ~\mbox{.}
\end{equation}
Explicit expressions for $G_z(E,n,m)$ are available from the Jost solutions $\psi_{\pm}(z,E)$, obtained by extending $(H_{x+i\epsilon,\pm}-E)^{-1}\delta_{\pm1} \in\ell^2(\Z_\pm)$ to satisfy the full line equation $H_{z}\psi=E\psi$. By construction, $\psi_\pm(z,E) \in \mathbb{C}^\mathbb{Z}$ does not have zeros, is $\mathit{l}^2$ at $\pm \infty$, and unique up to multiplicative constants. These solutions provide the formula
\begin{equation} \label{eq_greensfunction}
G_z(E, n,m)=\frac{\psi_-(z,E,\min\{m,n\})\psi_+(z,E;\max\{m,n\})}{W(\psi_-(z,E) ,\psi_+(z,E) )} ~\mbox{,}
\end{equation}
verified by direct computation. Here, $W(f,g) = f(n)g(n+1)-g(n)f(n+1)$ is the \emph{Wronskian}, which is $n$-independent (``conservation of the Wronskian'') if $f,g$ are both solutions to $H_{z}\psi=E\psi$. Similar computations for the half-line operators (eg. see \cite[\S1.2,2.1]{Teschl_book_2000}) show that
\begin{equation} \label{eq_sectionsJost}
m_{\pm}(z,E) = \frac{-\psi_\pm(z,E; \pm1)}{\psi_\pm(z,E;0)} ~\mbox{.}
\end{equation}
Finally, using the same argument as in the self-adjoint case, conservation of the Wronskian and unicity of the Jost solutions up to multiplicative constants show existence of $a(z, E)$ such that
\begin{equation} \label{eq_invariance}
\psi_\pm(z+\alpha, E; n) = a(z, E) \psi(z, E; n + 1) ~\mbox{.}
\end{equation}
We note that (\ref{eq_invariance}) will later imply invariance of $s_\pm(. + i \epsilon,E)$ under the action of the cocycle $(\alpha, B_\epsilon^E)$ (see (\ref{eq_invsplitting}) below).

To verify the angle formula, we first use (\ref{eq_greensfunction}) to express the right side of (\ref{eqn:goal}), thereby
\begin{align}\label{eqn:rhs}
|\langle\delta_0,(H_{z}-E)^{-1}\delta_0\rangle|^{-1} = \left|\frac{\psi_+(z,E;1)}{\psi_+(z,E;0)}-\frac{\psi_-(z,E;1)}{\psi_-(z,E;0)}\right| ~\mbox{.}
\end{align}

On the other hand, recasting $s_\pm(z,E)$ in terms of $\psi_\pm(z,E)$, we compute taking advantage of (\ref{eq_invariance}),
\begin{equation} \label{eqn:psi-plus-1-0}
s_+(z, E) = \dfrac{ \psi_+(z - \alpha, E; 0)}{\psi_+(z - \alpha, E; 1)} = \dfrac{ \psi_+(z, E; -1)}{\psi_+(z, E; 0)} = (E-v(z))-\frac{\psi_+(z,E;1)}{\psi_+(z,E;0)} ~\mbox{,}
\end{equation}
and
\begin{equation}\label{eqn:psi-minus-1-0}
s_-(z,E)=\frac{\psi_-(z,E;-1)}{\psi_-(z,E;0)}=(E-v(z))-\frac{\psi_-(z,E;1)}{\psi_-(z,E;0)} ~\mbox{.}
\end{equation}

Thus, combining (\ref{eqn:psi-plus-1-0}) and (\ref{eqn:psi-minus-1-0}) with (\ref{eqn:rhs}), we conclude (\ref{eqn:goal}), as claimed.
\end{proof}


Lemma \ref{lemma:angleformula} enters as the key ingredient in the spectral theoretic proof of Proposition \ref{prop_key}:
\begin{proof}[Spectral theoretic proof of Proposition \ref{prop_key}]
First observe that from (\ref{eq_invariance}), the sections $s_-(. + i\epsilon)$, $s_+(. + i \epsilon)$ are naturally $(\alpha, B_\epsilon^E)$-invariant: Under the identification $\mathbb{PC}^2\ni[(v_1, v_2)]\mapsto\frac{v_2}{v_1}\in\overline{\C}$, one concludes:
\begin{equation}\label{eqn:sinvariant}
B_\epsilon^E(x)\cdot s_{\pm}(z,E)=B_\epsilon^E(x) \begin{pmatrix}
\psi_\pm(z,E;0)\\\psi_\pm(z,E;-1)
\end{pmatrix}=\begin{pmatrix}
\psi_\pm(z,E;1)\\\psi_\pm(z,E;0)
\end{pmatrix}=s_{\pm}(z + \alpha,E) ~\mbox{.}
\end{equation}

By the angle formula,  $s_-(. + i \epsilon, E), s_+(. + i \epsilon, E)$ thus induce a continuous, $(\alpha, B_\epsilon^E)$-invariant splitting of $\mathbb{PC}^2$ expressed by the conjugacy,
\begin{equation} \label{eq_invsplitting}
C_\epsilon(x+\alpha)^{-1}B_\epsilon^E(x)C_\epsilon(x)=\begin{pmatrix}
\lambda_{1;\epsilon}(x)&0\\0&\lambda_{1; \epsilon}(x)^{-1}
\end{pmatrix} ~\mbox{.}
\end{equation}
where $\lambda_{1; \epsilon}(x)$ and
\begin{equation} \label{eq_conj}
C_\epsilon(x)=\frac{1}{\sqrt{1 - m_+(z - \alpha,E) m_-(z,E)}} \begin{pmatrix}1&m_+(z - \alpha,E)\\m_-(z,E)&1\end{pmatrix}    \in SL_2(\C)
\end{equation}
are continuous in $x$.

To conclude $(\alpha, B_\epsilon^E) \in \mathcal{UH}$, from (\ref{eq_condunifh}) it thus suffices to guarantee existence of $N \in \mathbb{N}$ such that uniformly in $x \in \mathbb{T}$,
\begin{equation} \label{eq_uh_condition_spth}
\prod_{j=0}^{N-1} \vert \lambda_{1; \epsilon}(x + j \alpha) \vert > 1 ~\mbox{.}
\end{equation}

(\ref{eq_uh_condition_spth}) follows immediately if we establish $L(\alpha, B_\epsilon^E) > 0$, in which case since $\psi_-(z, E)$ is $\mathit{l}^2$ at $-\infty$ and $E \in \rho(H_z)$, Oseledets' theorem determines 
\begin{equation} \label{eq_oseledets}
\frac{1}{n} \sum_{j =0}^{n-1} \log \vert \lambda_{1;\epsilon}(x+j \alpha) \vert \to L(\alpha, B_\epsilon^E) ~\mbox{.}
\end{equation}
Unique ergodicity of irrational rotations shows that the limit in (\ref{eq_oseledets}) is in fact {\em{uniform}} which yields (\ref{eq_uh_condition_spth}).

Finally, positivity of the complexified LE follows from a the following well-known growth Lemma, which dates back to the work of Sorets-Spencer (Proposition 1 in \cite{SoretsSpencer_1991}), see also \cite{Bourgain_book_2005}, Chapter 3. More recent generalizations to higher dimensional cocycles appeared in \cite{KleinDuarte_2013_posLE}, see Lemma 5.2 therein.
\begin{lemma}[Growth lemma]\label{lemma:growth}
For $n\in\N\cup\{0\}$, let $A_n:=\begin{pmatrix}a_n&-1\\1&0\end{pmatrix}$, where $|a_n|\ge\mu$, some $\mu>2$. Then for all $n\in\N$,
\begin{equation}\label{eqn:growth-lemma}
\frac{1}{n}\log\big\|\prod_{j={n-1}}^0A_j\big\|\ge\log(\mu-1)>0 ~\mbox{.}
\end{equation}
In particular, if $A_n=B_\epsilon^E(x+n\alpha)$ with $m(\epsilon; E)>2$, then $L(\alpha, B_\epsilon^E) >0$.
\end{lemma}
\end{proof}

As concluding remark we note that the proof of Lemma \ref{lemma:angleformula} can be adapted to show:
\begin{prop} \label{eq_uhangle}
Fix $\alpha$ irrational and $E \in \mathbb{R}$. Whenever $\epsilon$ is such that $(\alpha, B_\epsilon^E) \in \mathcal{UH}$, $E \in \rho(H_z)$ and the angle formula (\ref{eqn:goal}) holds.
\end{prop}
\begin{proof}
Clearly, using the dominating (unstable) and minoring (stable) sections, $(\alpha, B_\epsilon^E) \in \mathcal{UH}$ implies existence of linear independent solutions $\psi_\pm(z; E)$ of $H_z \psi = E \psi$ decaying exponentially at respectively $\pm \infty$. Invariance of the sections implies that $\psi_\pm(z; E)$ trivially satisfy (\ref{eq_invariance}). In summary, using these solutions in the formulae (\ref{eq_greensfunction}) shows that $E \in \rho(H_z)$. 

Define $m_\pm(z, E)$ from (\ref{eq_sectionsJost}). We note that one may have $m_\pm(z, E) = \infty$, since zeros of $\psi_\pm(z; E)$ are not excluded, however, the sections defined in (\ref{eq_definesection}),
\begin{equation}
s_-(z, E) = \frac{\psi_-(z, E; -1)}{\psi_-(z, E; 0)} ~\mbox{, } s_+(z, E) = \frac{\psi_+(z - \alpha, E; 0)}{\psi_+(z - \alpha, E; 1)} ~\mbox{,}
\end{equation}
cannot {\em{both}} be $\infty$. Indeed, $\psi_-(z, E; 0) = \psi_+(z - \alpha, E; 1)= 0$ would lead to $\psi_+(z, E; 0) = 0$ by (\ref{eq_invariance}), which in turn would imply zero Wronskian, thereby contradicting linear independence of $\psi_\pm(z; E)$. Thus the difference on the left hand side of (\ref{eqn:goal}) is well-defined in $\overline{\mathbb{C}}$. Now we can run through the rest of the argument in the proof of Lemma \ref{lemma:angleformula} to conclude (\ref{eqn:goal}). 
\end{proof}


\section{Some applications} \label{sec_someapplications}

We apply Theorem \ref{thm_main} to various model situations, starting with the uniform criterion in part (ii) of the theorem.

\subsection{Subcriticality uniformly on the spectrum} \label{sec_someapplications_unif}
To obtain an estimate for the uniform Herman radius, write (\ref{eq_potential}) in complex form, 
\begin{equation} \label{eq_potential_1}
v(x) = \sum_{1 \leq \vert j \vert\le M} \lambda_j e_j(x) ~\mbox{,}
\end{equation}
where $e_j(x) = \mathrm{e}^{2 \pi i j x}$ and
\begin{equation} \label{eq_Fouriercoeff_potential}
\lambda_j = a_j + \frac{b_j}{i} ~\mbox{, $j >0$ , } \lambda_{-j} = \overline{\lambda_j} ~\mbox{.}
\end{equation}
Then, 
\begin{eqnarray} \label{eq_lowerboundv}
\vert v(x+ i \epsilon) \vert & \geq & \vert \overline{ \lambda_{M} } e_{-M}(z) + \lambda_M e_M(z) \vert  - \sum_{1 \leq \vert j \vert \leq M-1} \vert \lambda_j \vert \vert e_j(z) \vert \nonumber \\
                                        & \geq & \vert \lambda_M \vert \mathrm{e}^{2 \pi M \epsilon} - \sum_{j=1}^{M-1} \vert \lambda_j \vert \mathrm{e}^{2 \pi j \epsilon} - \sum_{j=1}^{M} \vert \lambda_j \vert ~\mbox{.}
\end{eqnarray}
As $\epsilon \to + \infty$, the right-most side of (\ref{eq_lowerboundv}) will eventually be positive, in particular letting $y:= \mathrm{e}^{2 \pi \epsilon}$, $\epsilon_{H; \mathrm{unif}}$ can be estimated from above by the {\em{largest positive root}} $R_p$ of the polynomial
\begin{equation}\label{eqn:general-m1-poly}
p(y):=|\lambda_M|y^M-|\lambda_{M-1}|y^{M-1}-\cdots-|\lambda_1|y-\left(4+3\sum_{j=1}^M|\lambda_j|\right) ~\mbox{.}
\end{equation}
We note that $p(y)$ has a unique positive root (Decartes' rule of signs) and, since $p(1) < 0$, necessarily $R_p > 1$. Thus, $\epsilon_{H; \mathrm{unif}} \leq \frac{1}{2 \pi} \log R_p$ and Theorem \ref{thm_main} (ii) imply:
\begin{prop} \label{coro_condition}
All energies in the spectrum are subcritical if
\begin{equation} \label{eq_coro_condition}
\vert \lambda_M \vert^{1/(M-d)} R_p < 1 ~\mbox{,}
\end{equation}
where $R_p$ is the largest positive root of the polynomial $p(y)$ defined in (\ref{eqn:general-m1-poly}). 
\end{prop}
Identifying subcritical behavior hence reduces to finding $R_p$. 

As a first example, we consider the simplest nontrivial generalization of the AMO, letting $M=2$ in (\ref{eq_potential_1}). In this case, we can solve for $R_p$ exactly, giving
\begin{equation} 
R_p = \frac{|\lambda_1|}{2|\lambda_2|}+\frac{1}{2|\lambda_2|}\sqrt{|\lambda_1|^2+16|\lambda_2|+12|\lambda_1||\lambda_2|+12|\lambda_2|^2}, ~\mbox{.}
\end{equation}
The condition in Proposition \ref{coro_condition} thus yields subcritical behavior on all of the spectrum if
\begin{equation} \label{eqn:simple-m1-lambda-conditions}
|\lambda_1|+4|\lambda_2|+3|\lambda_2|^2+3|\lambda_1||\lambda_2|<1 ~\mbox{,} 
\end{equation}
which we illustrate in Fig. \ref{fig_simple_exact_m1}. 
\begin{figure}[h!]
\begin{center}
\includegraphics[scale=0.4]{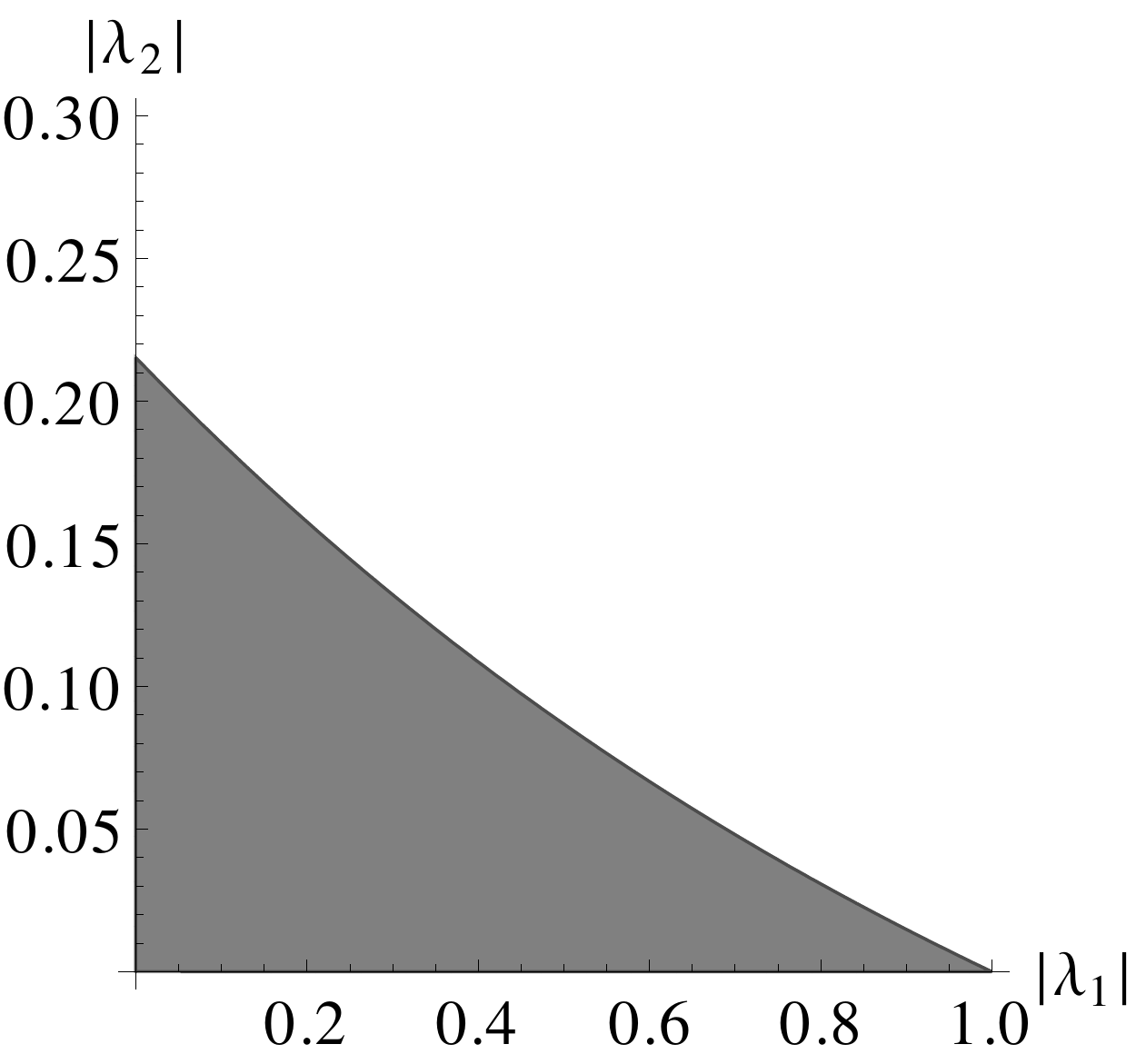}
\end{center}
\caption{The region of sub-criticality for $M=2$ determined in (\ref{eqn:simple-m1-lambda-conditions}). Note that the behavior captures the optimal boundary-value properties predicted by Fact \ref{fact_global} as $\lambda_2 \to 0$; in this limit, $\vert \lambda_1 \vert \to 1^-$.} \label{fig_simple_exact_m1}
\end{figure}

More generally, several articles in the physics literature consider the potential (e.g. \cite{HiramotoKohmoto_1989, SoukoulisEconomou_1982, ChaoRiklundLiu_1985,ChaoRiklundLiu_1986})
\begin{equation} \label{eq_genHarpers}
v(x) =  2 \mathrm{Re} \{ \lambda_1 e_1(x) + \lambda_M e_M(x)  \} ~\mbox{.}
\end{equation}
For $b_1 = b_M = 0$, this is known as {\em{generalized Harpers model}}, interesting also due to its relation to the quantized Hall effect in {\em{three}} dimensions  \cite{HiramotoKohmoto_1989, MontambauxKohmoto_1990}. 

In this case, upper bounds for the largest positive root of polynomials can be used to estimate $R_p$, for instance:
\begin{theorem}[\Stefanescu ~\cite{stefanescu1}] \label{thm:positive-root-bounds} 
Let $p(x)=x^d-b_1x^{d-m_1}-\cdots-b_kx^{d-m_k}+\sum_{j\ne m_1,\ldots,m_k}a_jx^{d-j}$, where $b_1,\ldots,b_k>0$ and $a_j\ge0$ for $j\not\in\{m_1,\ldots,m_k\}$. Then, 
\begin{equation}
B_1=\max\{(kb_1)^{1/{m_1}},\ldots,(kb_k)^{1/{m_k}} \}
\end{equation}
forms an upper bound for the positive roots of $p$.
\end{theorem}
Applying Theorem \ref{thm:positive-root-bounds} to estimate $R_p$ for the potential in (\ref{eq_genHarpers}), 
we conclude from Proposition \ref{coro_condition} that all energies in the spectrum are subcritical if
\begin{equation} \label{eqn:ghm-next-stefanescu}
2|\lambda_1|<1 \,\text{ and }\, 8|\lambda_M|^{1/(M-1)}+6|\lambda_1||\lambda_M|^{1/(M-1)}+6|\lambda_M|^{M/(M-1)}<1 ~\mbox{.} 
\end{equation}


\subsubsection{Limiting behavior for generalized Harper's model}
Finally, we analyze the limiting behavior produced by Proposition \ref{coro_condition} for the potential in (\ref{eq_genHarpers}). Observe that for $M=2$, (\ref{eqn:simple-m1-lambda-conditions}) and Figure \ref{fig_simple_exact_m1} explicitly show that as $\vert \lambda_2 \vert \to 0^+$, the region of subcriticality approaches $\vert \lambda_1 \vert \to 1^+$, as expected from the spectral properties of the AMO. Indeed, the same behavior follows more generally from Proposition \ref{coro_condition}, which we quantify in:
\begin{claim}
For all $|\lambda_1|<1$, there is $\mu>0$ so that for all $0<|\lambda_M|\le\mu$ and some $\kappa=\kappa(|\lambda_M|)$, Proposition \ref{coro_condition} guarantees subcritical behavior on all of the spectrum whenever $|\lambda_2|,\ldots,|\lambda_{M-1}|\le\kappa$. Specifically, writing $0 \leq \vert \lambda_1\vert = 1 - \delta_1 - \delta_2$, for $\delta_1, \delta_2 > 0$, one can take
\begin{equation}\label{eqn:amo-bounds}
\mu = \left( \frac{\delta_1}{4+3M} \right)^{M-1} ~\mbox{, } 
\end{equation}
and for $M>2$, $\kappa =\frac{\delta_2}{(M-2)} \left(\frac{\vert \lambda_M \vert}{2(M-2)}\right)^{M-2}$.
\end{claim}
\begin{proof}
Write $|\lambda_1|=1-\delta_1-\delta_2$ for $\delta_1,\delta_2>0$. Note that from (\ref{eqn:amo-bounds}) both $\mu, \kappa < 1$. Letting 
\begin{equation}
\lambda:=(M-2)\max_{2 \leq \vert j \vert \leq M-2} |\lambda_j| \le \kappa (M-2) ~\mbox{,}
\end{equation}
we estimate
\begin{eqnarray}
p(y) > |\lambda_M|y^M-y+(\delta_1 y-(4+3M))+(\delta_2 y-\lambda y^{M-1}) ~\mbox{.} \label{eqn:amo-pick-up-here}
\end{eqnarray}

\begin{lemma} \label{lemma_tedious}
Let $y_0:=\frac{1}{|\lambda_M|^{1/(M-1)}}$. Then $y\ge y_0$ implies $p(y) > 0$.
\end{lemma}
\begin{proof}
Let $y\ge y_0$. First note that (\ref{eqn:amo-bounds}) implies $y\ge\frac{4+3M}{\delta_1}$, so that 
\begin{equation}
(\delta_1 y-(4+3M))\ge0 ~\mbox{.}
\end{equation}

Denote by $z_0$ the unique positive root of $|\lambda_M|y^M-(M-2)y^{M-1}-y$. For $M=2$, $z_0 = y_0$. Using Theorem \ref{thm:positive-root-bounds}, for $M>2$ we estimate
\begin{equation} \label{eq_estimz0}
z_0 \leq \frac{2(M-2)}{\vert \lambda_M \vert} ~\mbox{.}
\end{equation}

We distinguish the cases $y\le z_0$ and $y>z_0$. If $y>z_0$, then $|\lambda_M|y^M-\lambda y^{M-1}-y \geq |\lambda_M|y^M-(M-2)y^{M-1}-y > 0$, so that the right side of (\ref{eqn:amo-pick-up-here}) is positive, whence the claim in the lemma is satisfied.

If $y\le z_0$, then for $M >2$, the definition of $\kappa$ and $\lambda$ as well as (\ref{eq_estimz0}) imply,  
\begin{equation} \label{eq_case}
(\delta_2 y-\lambda y^{M-1}) \geq \delta_2y\left(1-\left(\frac{z_0 \vert \lambda_M \vert}{2 (M-2)}\right)^{M-2}\right)\ge0 ~\mbox{.}
\end{equation}

Thus for $y \geq y_0 =\frac{1}{|\lambda_M|^{1/(M-1)}}$, (\ref{eqn:amo-pick-up-here}) yields
\begin{eqnarray}
p(y) > |\lambda_M|y^M-y+\underbrace{(\delta_1 y-(4+3M))}_{\ge0}+\underbrace{(\delta_2 y-\lambda y^{M-1})}_{\ge0}\ge|\lambda_M|y^M-y \geq 0 ~\mbox{,} \nonumber
\end{eqnarray}
as claimed.
\end{proof}

We conclude from Lemma \ref{lemma_tedious} that $R_p < y_0$, whence
\begin{equation}
R_p \vert \lambda_M \vert^{1/(M-1)} < 1 ~\mbox{,}
\end{equation}
which in particular implies the condition (\ref{eq_coro_condition}) in Claim \ref{coro_condition}.
\end{proof}

We note that qualitatively this behavior is expected from the point of view of Avila's global theory; it is an immediate consequence of upper-semicontinuity of the acceleration in the cocycle \cite{Avila_globalthy_published}:
\begin{fact} \label{fact_global}
Given a quasi-periodic Schr\"odinger operator $H_{v(x); \alpha_0}$ with analytic potential $v_0$ and irrational frequency $\alpha_0$, suppose that all energies in the spectrum of $H_{v_0(x); \alpha_0}$ are subcritical. Then the same is true for all $(v,\alpha) \in \mathcal{C}^\omega(\mathbb{T}; \mathbb{R}) \times \mathbb{R}\setminus \mathbb{Q}$ in some open neighborhood of $(v_0,\alpha_0)$. 
\end{fact}
\begin{proof}
By Theorem \ref{thm_complexLEglobal}, the acceleration for $\epsilon \geq 0$ can only attain non-negative integer values if $\alpha$ is irrational. Hence, upper semi-continuity of the acceleration and compactness of the spectrum ensure existence of some open neighborhood of $(\alpha_0, v_0)$ in $\mathcal{C}^\omega(\mathbb{T}; \mathbb{R}) \times \mathbb{R}\setminus \mathbb{Q}$, such that $\omega(\epsilon = 0, E) = 0$ for all $E$ in the spectrum. Here, we also use that the spectrum of quasi-periodic Schr\"odinger operators depends continuously on $(\alpha, v)$ in the Hausdorff metric \cite{AvronSimon_1983}. 
\end{proof}
We note that in spectral theoretic terms, Fact \ref{fact_global} states that purely ac spectrum is stable w.r.t. perturbations in $(\alpha, v)$.

Unfortunately, as $\lambda_1 \to 0^+$, Proposition \ref{coro_condition} does not capture the optimal behavior predicted by Fact \ref{fact_global}. Indeed, since $\widetilde{m}(\epsilon; E)$ depends continuously on the $\lambda_j$'s, we set $\lambda_1=0$ in (\ref{eq_genHarpers}) and explicitly solve for $\epsilon_{H;\mathrm{unif}}$ to find
\begin{equation}\label{eqn:badbehavior}
e^{2\pi\epsilon_{H;\mathrm{unif}}}=\left(\frac{2}{|\lambda_M|}+1+\sqrt{\frac{4}{|\lambda_M|^2}+\frac{4}{|\lambda_M|}+2}\right)^{1/M}\ge\left(\frac{4}{|\lambda_M|}\right)^{1/M} ~\mbox{,}
\end{equation}
which, if $\epsilon_{H;\mathrm{unif}}<-\frac{\log|\lambda_M|}{2\pi(M-1)}$, imposes the restriction
\begin{equation}\label{eqn:badbehavior-lambda-bound}
|\lambda_M|<\frac{1}{4^{\frac{M}{d}-1}}\le\frac{1}{4} ~\mbox{.}
\end{equation}
 
\subsection{Energy dependence} \label{sec_someapplications_energy}

To illustrate the application of Theorem \ref{thm_main} (i), we consider {\em{odd}} potentials where all $a_j = 0$ in (\ref{eq_potential}). In this case, it is known that $E=0 \in \Sigma(\alpha)$, indeed:
\begin{fact} \label{fact_symmetry}
Given a quasi-periodic Schr\"odinger operator such that for some $x_0 \in \mathbb{T}$, $v(. + x_0)$ is odd. Then, for every irrational $\alpha$, $0 \in \Sigma(\alpha)$.
\end{fact}
A proof of Fact \ref{fact_symmetry} can e.g. be found in \cite{BALASUBRAMANIAN_KULKARNI_RADHA_2001}; for the reader's convenience we give a slightly shorter, alternative argument in Appendix \ref{appendix_factsymmetry}.

As an example, we consider the simplest case where $v(x) = b_1 \sin(2 \pi x) + b_2 \sin(4 \pi x)$. Estimating like in (\ref{eq_lowerboundv}), allows to bound the Herman radius for $E=0$ from above by the largest positive root $R_q$ of the polynomial,
\begin{equation}
q(y) := \vert b_2 \vert y^2 - \vert b_1 \vert y - (\vert b_1 \vert + \vert b_2 \vert + 2) ~\mbox{,}
\end{equation}
where, as before, $y = \mathrm{e}^{2 \pi \epsilon}$. Thus, the condition (\ref{eq_subcritcrit_energy}) in Theorem \ref{thm_main} (i) asserts that $E=0$ is subcritical if $R_q \vert b_2 \vert < 1$, which, computing $R_q$, yields
\begin{equation} \label{eqn_exampleodd}
\vert b_2 \vert^2 + \vert b_1 \vert \vert b_2 \vert + 2 \vert b_2 \vert + \vert b_1 \vert < 1 ~\mbox{.}
\end{equation}
As mentioned earlier (see Remark \ref{remark_mainthm} (ii)), proving that $E=0$ is subcritical implies existence of some ac spectrum centered around $E=0$. This follows again using upper-semicontinuity of the acceleration, which since $E=0$ is subcritical, implies that $\omega(\epsilon=0, E) = 0$ for all $E$ in some interval $(E_1, E_2)$ containing $0$. Hence, by the almost reducibility theorem, all spectral measures are purely ac on $\Sigma(\alpha) \cap (E_1, E_2)$.

Fig. \ref{fig_some_ac_compare} depicts the region determined by (\ref{eqn_exampleodd}) where the operator has correspondingly {\em{some}} ac spectrum; the same figure compares this with the region determined by (\ref{eqn:simple-m1-lambda-conditions}) where {\em{all}} of the spectrum is purely ac continuous.
\begin{figure}[h!]
\begin{center}
\includegraphics[scale=0.4]{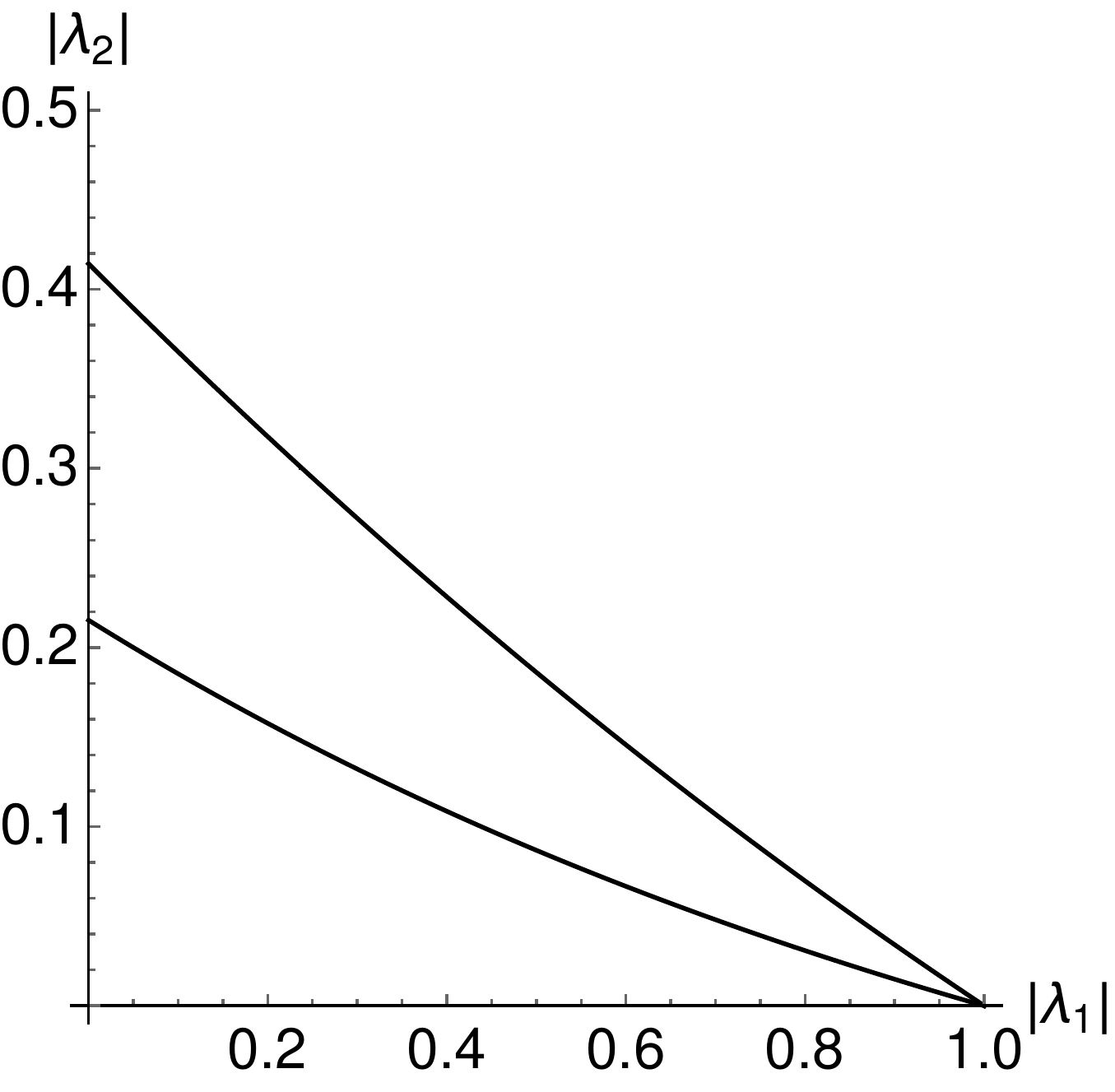}
\end{center}
\caption{For the Schr\"odinger operator with potential $v(x) = b_1 \sin(2 \pi x) + b_2 \sin(4 \pi x)$, the upper curve represents the boundary of the region determined by (\ref{eqn_exampleodd}); below the upper curve, $E=0$ is subcritical. This is compared with the region according to (\ref{eqn:simple-m1-lambda-conditions}) where {\em{all}} energies in the spectrum are subcritical; the lower curve represents the boundary of the latter.} \label{fig_some_ac_compare}
\end{figure}

More generally, using the same ideas for
\begin{equation} \label{eq_examplethm}
v(x) = 2 ( b_N \sin(2 \pi N x) + b_M \sin(2 \pi M x) ) ~\mbox{, $1 \leq N < M$, } b_M \neq 0 ~\mbox{,}
\end{equation}
we conclude, employing the root bound in Theorem \ref{thm:positive-root-bounds}, that:
\begin{ExampleTheorem}
Consider a quasi-periodic Schr\"odinger operator with potential given by (\ref{eq_examplethm}). $E = 0$ is subcritical if 
\begin{equation} \label{eq_examplethm_1}
\vert b_M \vert^{1/M-N} \max\left\{ \left(2 \frac{ \vert b_N \vert}{\vert b_M \vert}\right)^{1/(M-N)} ~,~ \left( \frac{4+ 2 \vert b_N \vert}{\vert b_M \vert} + 2 \right)^{1/M} \right\} < 1 ~\mbox{.}
\end{equation}
In particular, the operator has {\em{some}} ac spectrum if (\ref{eq_examplethm_1}) is met.
\end{ExampleTheorem}

\section{Jacobi operators} \label{sec_jacobi}

\subsection{Jacobi cocycles}
It is natural to try to extend Theorem \ref{thm_main} to Jacobi operators, 
\begin{equation} \label{eq_jacobiop}
(H_{c(x), v(x) ; \alpha}\psi)_{n} := \overline{c(x + (n-1) \alpha )}\psi_{n-1} + c(x +  n \alpha)\psi_{n+1} + v(x +  n \alpha)\psi_n ~\mbox{.}
\end{equation}
Here, the discrete Laplacian in (\ref{eq_qso}) is modified by evaluating the (complex) trigonometric polynomial
\begin{equation}
c(x) = \sum_{k=N_1}^{N_2} \mu_k \mathrm{e}^{2 \pi i k x} ~\mbox{, } N_1< N_2 ~\mbox{, } \vert \mu_{N_1} \mu_{N_2} \vert > 0 ~\mbox{,}
\end{equation}
along the trajectory $x \mapsto x + \alpha$. As before, $v$ is assumed to be (real) trigonometric polynomial of the form given in (\ref{eq_potential}).

A prominent example from physics is extended Harper's model (EHM) where both $c, v$ are both trigonometric polynomials of degree 1. Proposed by D. J. Thouless in context with the integer quantum Hall effect \cite{Thouless_1983}, EHM  generalizes the AMO, allowing for a wider range of lattice geometries by permitting the electrons to hop to both nearest and next nearest neighboring lattice sites. Its spectral theory has recently been solved \cite{AvilaJitomirskayaMarx_preprint_2015}, relying on an extension of parts of Avila's global theory to analytic Jacobi operators \cite{JitomirskayaMarx_2012, JitomirskayaMarx_2013_erratum}. In fact, the ``method of almost constant cocycles'' underlying the complex Herman formula was originally developed in \cite{JitomirskayaMarx_2012} to find the complexified LE for EHM.

As has been detailed in \cite{JitomirskayaMarx_review_2015, Marx_thesis}, the complexified LE for Jacobi operators is defined by\footnote{Replacing the complex conjugate of $c$ by its reflection along the real line as done in the upper right corner of (\ref{eq_jacobicocycle}) makes (\ref{eq_jacobicocycle}) an analytic, matrix-valued function.}
\begin{equation} \label{eq_complexle_jacobi}
L(\epsilon; E) := L(\alpha, A_\epsilon^E) - I(c) ~\mbox{, } I(c) := \int_\mathbb{T} \log \vert c(x) \vert \ud x ~\mbox{,}
\end{equation}
where $L(\alpha, A_\epsilon^E)$ is the LE of the {\em{phase-complexified Jacobi cocycle}} induced by 
\begin{equation} \label{eq_jacobicocycle}
A_\epsilon^E(x) := \begin{pmatrix}  E - v(x + i \epsilon) & - \overline{c(x - i \epsilon - \alpha)} \\ c(x + i \epsilon) & 0 \end{pmatrix} ~\mbox{.}
\end{equation}
Since $I(c)$ is a constant independent of $\epsilon$, determining the complexified LE (\ref{eq_jacobicocycle}) reduces to finding the LE associated with (\ref{eq_jacobicocycle}).

As before, letting $\epsilon = 0$ in (\ref{eq_complexle_jacobi}) yields what is usually known as {\em{the}} LE of a Jacobi operator (\ref{eq_jacobiop}). We mention that alternative choices for Jacobi cocycles exist (see e.g. \cite{Marx_2014, JitomirskayaMarx_review_2015} for details), however for what is to come, (\ref{eq_jacobicocycle}) turns out to be the most advantageous.

One feature not present for Schr\"odinger cocycles is that $(\alpha, A_\epsilon^E)$ is in general non-invertible, indeed, $\det A_\epsilon^E(x) = \overline{c(x - i \epsilon - \alpha)} c(x + i \epsilon)$ which may vanish due to zeros of $c$. A Jacobi operator is called {\em{singular}} if $c(x)$ has zeros on $\mathbb{T}$ and {\em{non-singular}} otherwise. Singularities of the cocycle often lead to interesting phenomena when trying to generalize results for Schr\"odinger operators to the Jacobi case, which has been explored in several recent articles \cite{JitomirskayaMarx_review_2015, Marx_2014, taovoda, AvilaJitomirskayaSadel_2013, bindervoda, BinderVoda_2013, JitomirskayaMarx_2012, KTao_2012, KTao_preprint_2011, JitomirskayaMarx_2011}. From a dynamical point of view, presence of singularities is accounted for by replacing uniform hyperbolicity ($\mathcal{UH}$) with uniform domination ($\mathcal{DS}$) (recall Sec. \ref{sec_uniformds} for a definition)  \cite{Marx_2014, AvilaJitomirskayaSadel_2013}.

It was proven in \cite{JitomirskayaMarx_2012, JitomirskayaMarx_2013_erratum} (see also \cite{JitomirskayaMarx_review_2015}) that Theorem \ref{thm_complexLEglobal} essentially carries over to Jacobi operators:
\begin{theorem} \label{thm_complexLEglobal_jacobi}
\begin{itemize}
The analytic properties of $\epsilon \mapsto L(\epsilon; E)$ stated in Theorem \ref{thm_complexLEglobal} hold essentially unchanged with the only alteration that $\omega(\epsilon; E) \in \frac{1}{2} \mathbb{Z}$. For {\em{non}}-singular Jacobi operators one still has $\omega(0; E) \in \mathbb{Z}$, in particular the smallest {\em{non-zero}} value of $\omega(0; E)$ is 1.
\end{itemize}
\end{theorem}

In particular, for both singular and non-singular Jacobi operators, Fig. \ref{fig:complexified-LE-energies} represents the three possible situations for the graph of $L(\epsilon; E)$ in a neighborhood of $\epsilon = 0$ if $E \in \Sigma$. From a spectral theoretic point of view the implications for each of these three cases for {\em{non-singular}} operators are the same as in the Schr\"odinger case \cite{AvilaJitomirskayaMarx_preprint_2015, JitomirskayaMarx_review_2015}; one thus partitions $\Sigma$ into subcritical, supercritical, and critical energies. 

For {\em{singular}} operators, partitioning the set $\mathcal{Z} := \{E: ~L(0; E) = 0\}$ into subcritical and critical does not yield any further information, as the presence of zeros of $c(x)$ a priori excludes any absolutely continuous spectrum for all $x \in \mathbb{T}$ \cite{Dombrowsky_1978}. 
For singular Jacobi operators, the criteria derived below will thus simply imply zero LE.

\subsection{Asymptotic Analysis}

We now turn to generalizing Theorem \ref{thm_main} to Jacobi operators, and will encounter two complications:

First, observe that the method of almost constant cocycles does not generalize immediately to the Jacobi case - if we naively factor out the leading terms in $A^E_\epsilon$ in analogy to (\ref{eq_almostconst}),  the remainder for the case $\mathrm{deg}(c) > \mathrm{deg}(v)$ will in general approach a constant matrix with {\em{zero}} spectral radius (and thus LE $= - \infty$), which would yield an undetermined expression for $L(\alpha, A_\epsilon^E)$ in the limit $\epsilon \to\infty$. 

This problem is remedied by first conjugating $A_\epsilon^E$ by 
\begin{equation} \label{eq_conjuacyjacobi}
C = \begin{pmatrix} 1 & 0 \\ 0 & w^K \end{pmatrix} ~\mbox{,}
\end{equation}
for some appropriate $K \in \frac{1}{2} \mathbb{Z}$, where for convenience we write $w:= \mathrm{e}^{- 2 \pi i (x + i \epsilon)}$. Here, we call two cocycles $(\alpha, A)$ and $(\alpha, A^\prime)$ with $A, A^\prime \in \mathcal{C}^\omega(\mathbb{R}/ 2 \mathbb{Z}, M_2(\mathbb{C}))$ conjugate if $A^\prime(x) = C(x + \alpha)^{-1} A(x) C(x)$ for some $C \in \mathcal{C}^\omega(\mathbb{R} / 2 \mathbb{Z}, GL(2, \mathbb{C}))$. Conjugacies clearly preserve the LE.

Applying the conjugacy in (\ref{eq_conjuacyjacobi}) will lead to consideration of cases, depending on the sign of $2M - (N_2 - N_1)$. The latter expresses the dependence of the asymptotics of $L(\epsilon; E)$ on the relative degree of $c$ and $v$. We note that from a spectral theoretic point of view, the conjugacy in (\ref{eq_conjuacyjacobi}) is equivalent to a known unitary whose action transforms the original Jacobi operator to one where $c(x)$ is replaced by $c(x) \mathrm{e}^{- 2 \pi i K x}$, see e.g. \cite{Teschl_book_2000}, Lemma (1.57) and Lemma 1.6, therein. 

The second complication is of fundamental nature. The key in Sec. \ref{sec_uniformds} was to {\em{quantify}} when the asymptotics expressed through the complex Herman formula holds. This was possible since Schr\"odinger cocycles are asymptotically (as $\epsilon \to \infty$) close to the constant $D_\infty = (\begin{smallmatrix} 1 & 0 \\ 0 & 0 \end{smallmatrix})$ with $(\alpha, D_\infty) \in \mathcal{DS}$; quantifying the asymptotics then amounted to finding the radius of stability of $\mathcal{DS}$ about $(\alpha, D_\infty)$. 

Depending on the sign of $2M - (N_2 - N_1)$, this will not be possible for Jacobi operators. Indeed analyzing the case $2M - (N_2 - N_1) < 0$ will show that  the limiting constant cocycle lies on the \emph{boundary} of $\mathcal{DS}$. The asymptotic expression however still leads a Herman bound which, due to above remarks, is not entirely trivial and has not appeared before.

We turn to the analysis of the cases, starting with $2M > (N_2 - N_1)$, where a criterion for sub-criticality can be obtained. Here, the analogue of (\ref{eq_defnminfunc}) is defined by
\begin{equation} \label{eq_jacobi_auxfun}
m(\epsilon; E) := \min_{x \in \mathbb{T}} \dfrac{ \vert E - v(x + i \epsilon) \vert^2 }{\vert c(x+ i \epsilon) c(x - i \epsilon - \alpha) \vert} ~\mbox{,}
\end{equation}
respectively, uniformly over $E \in \Sigma$,
\begin{equation} \label{eq_jacobi_auxfun_unif}
\widetilde{m}(\epsilon; E) := \min_{x \in \mathbb{T}} \dfrac{ \left( \vert v(x + i \epsilon) \vert - 2 \left( \sum_{j = N_1}^{N_2} \vert \mu _j \vert + \sum_{j=1}^{M} \vert \lambda_j \vert \right) \right)^2 }{\vert c(x+ i \epsilon) c(x - i \epsilon - \alpha) \vert} ~\mbox{,}
\end{equation}
where $\lambda_j$ is given in (\ref{eq_Fouriercoeff_potential}), as earlier.

Define the corresponding Herman radii, $\epsilon_H$ and $\epsilon_{H; \mathrm{unif}}$, as the largest $\epsilon \geq 0$ such that, respectively, $m(\epsilon; E) = 4$ and $\widetilde{m}(\epsilon; E) = 4$. 

\begin{theorem} \label{thm_jacobi_case1}
Suppose $2M > (N_2 - N_1)$.
\begin{itemize}
\item[(i)] Then for every $E \in \mathbb{R}$,
\begin{equation} \label{eq_case1_complexherman}
L(\alpha, A^E) = 2 \pi K \vert \epsilon \vert + \log \vert \lambda_M \vert ~\mbox{, all $\vert \epsilon \vert \geq \epsilon_H$.}
\end{equation}
In particular one has the Herman bound,
\begin{equation} \label{eq_case1_hermanbd}
L(0; E) \geq  \log \vert \lambda_M \vert - I(c)  ~\mbox{.}
\end{equation}

\item[(ii)]
For {\em{non-singular}} Jacobi operators, $E \in \Sigma(\alpha)$ is subcritical whenever
\begin{equation} \label{eq_case1_subcr}
\epsilon_H < \dfrac{-\log \vert \lambda_M \vert + I(c)}{2 \pi (M-1)} ~\mbox{.}
\end{equation}

All of the spectrum is subcritical whenever 
\begin{equation} \label{eq_case1_subcrunif}
\epsilon_{H; \mathrm{unif}} < \dfrac{-\log \vert \lambda_M \vert + I(c)}{2 \pi (M-1)} ~\mbox{.}
\end{equation}
\end{itemize}
For {\em{singular}} Jacobi operators, (\ref{eq_case1_subcr}) and (\ref{eq_case1_subcrunif}) simply imply $L(0; E) = 0$.
\end{theorem}

\begin{proof}
Let $K \in \frac{1}{2} \mathbb{Z}$ to be determined later. Conjugating the Jacobi cocycle by $C$ in (\ref{eq_conjuacyjacobi}), we obtain 
\begin{equation} \label{eq_jacobiafterconj}
A^\prime_\epsilon(x) :=  \begin{pmatrix} E - v(x+ i \epsilon) & -w^K \overline{c(x - i \epsilon -\alpha)} \\ z^{-K} e^{2\pi i K\alpha} c(x + i \epsilon) & 0 \end{pmatrix} ~\mbox{.}
\end{equation}
Note that taking the limit $\epsilon \to + \infty$ is equivalent to $\vert w \vert \to \infty$. Thus, expressing the two off-diagonal terms of (\ref{eq_jacobiafterconj}) in terms of $w$,
\begin{eqnarray}
 w^K \overline{c(x - i \epsilon -\alpha)} &=& \sum_{j = N_1}^{N_2} \overline{\mu_j}w^{j+K}e^{2\pi i j \alpha}  ~\mbox{,}    \\
w^{-K} e^{2\pi i K \alpha}c(x + i \epsilon) &=& e^{2 \pi i M \alpha} \sum_{j = N_1}^{N_2} \mu_j w^{-(j+K)} ~\mbox{,}
\end{eqnarray}
we see that the upper left corner in (\ref{eq_jacobiafterconj}) will dominate as $\epsilon \to + \infty$, if
\begin{eqnarray} \label{eq_cond_potdom}
N_2 + K < M ~\mbox{,} \nonumber \\
N_1 + K > - M ~\mbox{,} 
\end{eqnarray}
which is possible since $N_2 - N_1 < 2M$. 

In agreement with (\ref{eq_cond_potdom}), we pick $K = - \frac{N_1 + N_2}{2}$.
Now we can factor out the dominating term in $A^\prime_{\epsilon}$ and take $\epsilon \to + \infty$, giving
\begin{equation} \label{mat} 
L(\alpha, A_{\epsilon}^E) = L(\alpha, A_\epsilon^\prime) = 2 \pi M \epsilon + \log{|\lambda_M|}+ L(\alpha, \begin{pmatrix} 1 + f_a(\epsilon, x) & f_b(\epsilon, x) \\ f_c(\epsilon, x) & 0  \end{pmatrix} ) 
\end{equation}
where
\begin{eqnarray}
 f_a(\epsilon, x) =  \dfrac{E - v(x + i \epsilon)}{\lambda_M e^{2\pi M \epsilon - 2\pi i M x}} - 1 ~\mbox{, }
& |f_b(\epsilon, x)| =  \dfrac{|c(x - i \epsilon - \alpha)|}{|\lambda_M| e^{2 \pi (M - K) \epsilon}}~\mbox{, } \\
& |f_c(\epsilon, x)| = \dfrac{|c(x + i \epsilon)|}{|\lambda_M| e^{2 \pi (M + K) \epsilon}} ~\mbox{.}
\end{eqnarray}

Since $N_2 - N_1 < 2M$, $f_a, f_b, f_c = o(1)$ uniformly in $x$ as $\epsilon \to +\infty$, hence
\begin{equation}
(\alpha, D_\epsilon): = (\alpha, \begin{pmatrix} 1 + f_a(\epsilon, x) & f_b(\epsilon, x) \\ f_c(\epsilon, x) & 0  \end{pmatrix}) \to (\alpha, \begin{pmatrix} 1 & 0 \\ 0 & 0 \end{pmatrix}) \in \mathcal{DS} ~\mbox{.}
\end{equation}
Thus, from the stability statement in Lemma \ref{lem_stability}, we conclude $(\alpha, D_\epsilon) \in \mathcal{DS}$ if $\epsilon > \epsilon_H$.

On the other hand it is known that \cite{AvilaJitomirskayaSadel_2013} \footnote{Indeed, here we only use the ``only if'' direction which is essentially trivial.}, $L(\epsilon; E)$ is linear and positive if and only if $(\alpha, A_\epsilon^E) \in \mathcal{DS}$, which, combining Theorem \ref{thm_complexLEglobal_jacobi} and \ref{mat}, yields (\ref{eq_case1_complexherman}). Note that as in the Schr\"odinger case, we use that $\epsilon \mapsto m(\epsilon; E)$ increases strictly for $\epsilon \geq \epsilon_H$. Finally, convexity of the complexified LE implies the Herman bound, (\ref{eq_case1_hermanbd}).

To prove part (ii), we first assume that the Jacobi operator is {\em{non-singular}} and follow the same contrapositive argument as in the proof of Theorem \ref{thm_main}. If $E \in \Sigma$ is {\em{not}} subcritical, $\omega(\epsilon =0 ; E) \geq 1$, which by convexity would imply the upper bound
\begin{equation} \label{eq_jacobi_case1_upperbd}
0 \leq L(\epsilon; E) \leq \log| \lambda_M |+2\pi\epsilon_H+2\pi (\epsilon-\epsilon_H) ~\mbox{, $0 \leq \epsilon \leq \epsilon_H$.}
\end{equation}
Notice, that by Theorem \ref{thm_complexLEglobal_jacobi}, $\omega(0; E) \geq 1$ for non-singular operators if $E \in \Sigma$ is not subcritical\footnote{For singular operators, the least positive value would be $\frac{1}{2}$ which would immediately imply a weaker form of (\ref{eq_jacobi_case1_upperbd}). The limiting argument below however allows to improve on that.}.

In particular, letting $\epsilon = 0$, we obtain (\ref{eq_case1_subcr}) upon taking the contrapositive. Finally, the uniform condition in (\ref{eq_case1_subcrunif}) follows immediately estimating the spectral radius of $H_{c(x), v(x); \alpha}$ from above by
\begin{equation}
2 \Vert c \Vert_\infty + \Vert v \Vert_\infty \leq 2 \left( \sum_{j = N_1}^{N_2} \vert \mu _j \vert + \sum_{j=1}^{M} \vert \lambda_j \vert \right) ~\mbox{.}
\end{equation}

Finally, if the operator is {\em{singular}}, we can use density of non-singular Jacobi operators in operator-norm topology (see Lemma \ref{Lem density}, below) to extend the upper bound in (\ref{eq_jacobi_case1_upperbd}) to the singular case, which then yields the claim in (ii): To see this, by the proof of Lemma \ref{Lem density},  there exists a sequence $(c_n)$ of trigonometric polynomials such that $c_n \to c$ in $\mathcal{C}^\omega(\mathbb{T})$, and for all $n \in \mathbb{N}$, $c_n$ has upper and lower degrees $N_2$ and $N_1$, respectively, and has no zeros on $\mathbb{T}$. Note that this in particular implies that the condition $2M > (N_2 - N_1)$ holds along the resulting sequence of approximating non-singular Jacobi operators, which allows application of above argument for the non-singular case. 

Denote the spectrum of $H_{c_n(x), v(x); \alpha}$ by $\Sigma_n$ and the spectrum of $H_{c_(x), v(x); \alpha}$ by $\Sigma$. Clearly, for all $x \in \mathbb{T}$, $H_{c_n(x), v(x); \alpha} \to H_{c_(x), v(x); \alpha}$ in norm-topology, in particular, $\Sigma_n \to \Sigma$ in the Hausdorff metric.

Suppose $E \in \Sigma$ does {\em{not}} satisfy $L(\epsilon = 0; E) = 0$. Then, taking $E_n \in \Sigma_n$ such that $E_n \to E$, continuity of the LE of analytic cocycles \cite{JitomirskayaMarx_2012}, implies that $L(\epsilon = 0; E_n) > 0$, eventually; in particular, $E_n$ is not subcritical whence (\ref{eq_jacobi_case1_upperbd}) holds eventually. Taking the limit implies that the upper bound in (\ref{eq_jacobi_case1_upperbd}) holds for $E$. 
\end{proof}

\begin{lemma}\label{Lem density}
The set of functions in $\mathcal{C}^{\omega}(\T)$ which are bounded away from zero on $\mathbb{T}$ is open and dense in $\mathcal{C}^{\omega}(\T)$. In particular, non-singularity of analytic Jacobi operators is Baire-generic in operator norm.
\end{lemma}

\begin{proof}
Openness is clear. To show density, given $0 \not \equiv f \in  \mathcal{C}^{\omega}(\T)$ that has zeros on $\mathbb{T}$, factorize $f(x) = t(x)\cdot q(x)$, where $q$ is zero-free on $\mathbb{T}$ and  $t(x)$ is a trigonometric polynomial containing all the zeros of $f$ on $\mathbb{T}$. Then let $(\epsilon_n)$ be a real sequence with $\epsilon_n = o(1)$, letting
\begin{equation}
f_n(x) : = t(x+i\epsilon_n)q(x) ~\mbox{,}
\end{equation}
$f_n$ has no zeros on $\mathbb{T}$ and $f_n \to f$ in $\mathcal{C}^\omega(\mathbb{T})$.

We note that if $f$ is a trigonometric polynomial, then $f_n$ is a trigonometric polynomial of the same degree than $f$; the latter is relevant for the proof of Theorem \ref{thm_jacobi_case1} (ii).
\end{proof}

Note that in order to apply the stability Lemma \ref{lem_stability}, the previous argument relied on $m(\epsilon; E) \to +\infty$ as $\epsilon \to +\infty$, in particular $m(\epsilon; E) = 4$ eventually. If $2M = (N_2 - N_1)$, this is not the case anymore, in fact
\begin{equation}
m(\epsilon; E) \to \dfrac{\vert \mu_{N_1} \mu_{N_2} \vert}{\vert \lambda_M \vert^2} ~\mbox{, as $\epsilon \to + \infty$ .}
\end{equation}
Lemma \ref{lem_stability} is however still applicable if $\frac{\vert \mu_{N_1} \mu_{N_2} \vert}{\vert \lambda_M \vert^2} < \frac{1}{4}$, in which case both $\epsilon_H$ and $\epsilon_{H; \mathrm{unif}}$ are well-defined. Hence, we conclude:
\begin{theorem}
Suppose $2M = N_2 - N_1$.
\begin{itemize}
\item[(i)] For all $E \in \mathbb{R}$ and {\em{some}} $0 \leq \epsilon_0$,
\begin{equation} \label{asymp2intro}
 L(\alpha, A_{\epsilon}^E) = 2 \pi M | \epsilon| + \log{\left( \max_\pm \left| \dfrac{\lambda_M \pm \sqrt{\lambda_M^2 - 4\overline{\mu_{N_2}}\mu_{N_1}e^{2\pi i M \alpha}}}{2}\right| \right)} ~\mbox{, all $|\epsilon| \geq \epsilon_0$.}
\end{equation}
In particular, one has the Herman bound
\begin{equation} \label{eq_case2_hermanbd}
L(0; E) \geq \log{\left( \max_\pm \left| \dfrac{\lambda_M \pm \sqrt{\lambda_M^2 - 4\overline{\mu_{N_2}}\mu_{N_1}e^{2\pi i M \alpha}}}{2}\right| \right)} - I(c) ~\mbox{.}
\end{equation}

If
\begin{equation} \label{eq_case2_cond}
\frac{|\mu_{N_1}\mu_{N_2}|}{|\lambda_M|^2} < \frac{1}{4} ~\mbox{,}
\end{equation}
then (\ref{asymp2intro}) holds for all $\vert \epsilon \vert \geq \epsilon_H$.

\item[(ii)] Suppose (\ref{eq_case2_cond}) holds. 

If the operator is {\em{non-singular}}, then $E \in \Sigma(\alpha)$ is subcritical whenever
\begin{equation}  \label{eq_case2_subcr}
\epsilon_H < \frac{-\log{\left( \max_\pm \left| \dfrac{\lambda_M \pm \sqrt{\lambda_M^2 - 4\overline{\mu_{N_2}}\mu_{N_1}e^{2\pi i M \alpha}}}{2}\right| \right)} + I(c)}{2\pi(M-1)} ~\mbox{.}
\end{equation}
All of the spectrum is subcritical whenever (\ref{eq_case2_subcr}) holds with $\epsilon_H $ replaced by $\epsilon_{H; \mathrm{unif}}$.

For {\em{singular}} Jacobi operators, (\ref{eq_case2_cond}) and (\ref{eq_case2_subcr}) imply $L(0; E) = 0$.
\end{itemize}
\end{theorem}
\begin{remark}
We note the frequency dependence in the lower bound (\ref{eq_case2_hermanbd}). Indeed, it follows from our earlier work on the LE of extended Harper's model in \cite{JitomirskayaMarx_2012}, that for $N_1 = -1, N_2 = 1$ and $M=1$, the frequency dependence of the asymptotics (\ref{asymp2intro}) persists as $\epsilon \to 0$, resulting in a frequency dependence of the LE of the Jacobi operator,
\begin{equation}
L(0; E) = \max\left\{ \log{\left( \max_{\pm} \left| \dfrac{\lambda_1 \pm \sqrt{\lambda_1^2 - 4\overline{\mu_{1}}\mu_{-1}e^{2\pi i \alpha}}}{2}\right| \right)} - I(c); 0 \right\} ~\mbox{.}
\end{equation}
This is interesting, since for quasi-periodic Schr\"odinger operators there is no known example of a LE with explicit dependence on $\alpha$.
\end{remark}

\begin{proof}
The argument follows the same steps as in the proof of Theorem \ref{thm_jacobi_case1}. We again conjugate by $C= (\begin{smallmatrix} 1 & 0 \\  0 & w^K \end{smallmatrix})$, this time with $K = M - N_2 = -N_1 - M$. The leading terms in the off-diagonal entries of (\ref{eq_jacobiafterconj}) are now $\overline{\mu_{N_2}}w^{M}e^{2\pi i N_2 \alpha}$ and $e^{2 \pi i K \alpha} \mu_{N_1} w^{M}$. Thus, we can pull out $w^M$ and use Theorem \ref{thm_complexLEglobal_jacobi} to see that for $\vert \epsilon \vert$ sufficiently large, 
\begin{eqnarray} 
L(\alpha, A_{\epsilon}^E) &=& 2\pi M |\epsilon| + \log{ \rho \begin{pmatrix} \lambda_M &  \overline{\mu_{N_2}}e^{2\pi i N_2 \alpha} \\ e^{2 \pi i (M-N_2) \alpha} \mu_{N_1} & 0 \end{pmatrix} } \label{eq_jacobi_case2_1} \\
&=& 2 \pi K | \epsilon| + \log{\left( \max_\pm \left| \dfrac{\lambda_M \pm \sqrt{\lambda_M^2 - 4\overline{\mu_{N_2}}\mu_{N_1}e^{2\pi i M \alpha}}}{2} \right| \right)} ~\mbox{,} \label{eq_jacobi_case2_2}
\end{eqnarray} 
which by convexity implies the lower bound in (\ref{eq_case2_hermanbd}). In (\ref{eq_jacobi_case2_1}),  we use $\rho(M)$ to denote the spectral radius of a matrix $M$. 

Moreover, if (\ref{eq_case2_cond}) holds, the limiting constant cocycle in (\ref{eq_jacobi_case2_1}) induces a dominated splitting, whence we conclude from Lemma \ref{lem_stability} that (\ref{eq_jacobi_case2_2}) holds for $\vert \epsilon \vert \geq \epsilon_H$. Here, we also use that concavity of $\epsilon \mapsto \log m(\epsilon;E)$ (Hadamard's three-circle theorem) and the maximum modulus principle imply that $m(\epsilon; E) \nearrow \frac{|\mu_{N_1}\mu_{N_2}|}{|\lambda_M|^2}$ strictly as $\epsilon_H \leq \epsilon \to + \infty$.

Part (ii) of the theorem is now obtained using identical arguments as in the proof of Theorem \ref{thm_jacobi_case1}.
\end{proof}

Lastly, we turn to the case when $2M - (N_2 - N_1) < 0$. Then, the conjugacy mediated by $C$ in (\ref{eq_conjuacyjacobi}) does not resolve the problem of an undetermined expression when considering the asymptotics. 

Instead, we consider the second iterate of the Jacobi-cocycle, $(2\alpha, A_{\epsilon; 2}^E)$, where $A_{\epsilon; 2}^E = A_\epsilon^E(x + \alpha) A_\epsilon^E(x)$ with
\begin{eqnarray}
A_{\epsilon; 2}^E(x + i \epsilon) = \begin{pmatrix} (E-v(x))(E-v(x + \alpha)) - \overline{c(x - i \epsilon)} c(x)  & -(E-v(x+ \alpha))\overline{c(x - i \epsilon - \alpha)} \\ c(x+\alpha)(E-v(x)) & -\overline{c(x - i \epsilon - \alpha)}c(x + \alpha) \end{pmatrix} ~\mbox{,} \nonumber
\end{eqnarray} 
and $L(2\alpha, A_{\epsilon; 2}^E) = 2 L(\alpha, A_\epsilon^E)$. Then, since $2M < N_2 - N_1$, we can write
\begin{equation}
A_{\epsilon; 2}^E(x + i \epsilon) = \mathrm{e}^{2 \pi (N_2 - N_1) \epsilon} \mathrm{e}^{-2 \pi i (N_2 - N_1) x} \begin{pmatrix} o(1) - \mu_{N_1}\overline{\mu_{N_2}} & o(1) \\  o(1) & o(1) - \mu_{N_1}\overline{\mu_{N_2}} \end{pmatrix} ~\mbox{,}
\end{equation}
as $\epsilon \to +\infty$.

By Theorem \ref{thm_complexLEglobal_jacobi}, we thus conclude:
\begin{theorem}
If $N_2 - N_1 > 2M$, then for all $E \in \mathbb{R}$, there exists $0 \leq \epsilon_0$ such that
\begin{equation}
L(\alpha, A_{\epsilon}^E) = \pi(N_2 - N_1)|\epsilon| + \frac{1}{2}(\log{|\mu_{N_1}|} + \log{|\mu_{N_2}|}) ~\mbox{all $\vert \epsilon \vert \geq \epsilon_0$.} 
\end{equation}
In particular, one has the Herman bound,
\begin{equation}
L(0; E) \geq \frac{1}{2}(\log{|\mu_{N_1}|} + \log{|\mu_{N_2}|}) - I(c) ~\mbox{.}
\end{equation}

\end{theorem}

\section{Some remarks on supercritical behavior} \label{sec_supercrit}

In this final section, we comment briefly on how one could use ideas from Sec. \ref{sec_uniformds} to obtain conclusions about {\em{supercritical}} behavior (i.e. positivity of the LE) for quasi-periodic Schr\"odinger operators. The following relies on the estimates of the complexified LE in Proposition \ref{prop_key_amended}, which, assuming existence of some $\epsilon_1$ satisfying 
\begin{equation} \label{eq_epsilon1}
0 \leq \epsilon_1 < \epsilon_H ~\mbox{with } m(\epsilon; E) > 2 ~\mbox{.}
\end{equation}
will allow to extract a lower bound for $L(\alpha, B^E)$ (see (\ref{eq_imrpovedHerman}), below), thereby improving on the classical Herman bound (\ref{eq_herman}). 

Testing for existence of such $\epsilon_1$ requires estimates of the function $m(\epsilon; E)$ {\em{outside}} the asymptotic regime which, unfortunately, we have found difficult to extract. On the other hand, it is easy to solve for $m(\epsilon; E)$ numerically, which at least gives rise to a simple numerical scheme to test for supercritical behavior. Below, we will demonstrate this for generalized Harper's model (\ref{eq_genHarpers}) with $M=2$.

Assuming existence of $\epsilon_1$ satisfying (\ref{eq_epsilon1}), we first establish above mentioned improvement on the Herman bound:
\begin{prop} \label{prop_positiveLE}
Consider a quasi-periodic Schr\"odinger operator with trigonometric potential, $v(x) = \sum_{j = -M}^{M} \lambda_j \mathrm{e}^{2 \pi i j x}$, $\lambda_j = \overline{\lambda_{-j}}$ and $\alpha$ irrational. Given $E \in \mathbb{R}$, suppose there is $\epsilon_1 = \epsilon_1(E; \{ \lambda_j \}_{j=1}^{M}\}$ satisfying (\ref{eq_epsilon1}). Then
\begin{equation}\label{eqn:LE-DS-ineq-ii}
L(\epsilon;E)\geq \log|\lambda_M|+2\pi M\epsilon+\gamma\left(\frac{\epsilon_H-\epsilon}{\epsilon_H-\epsilon_1}\right) ~\mbox{, } 0\leq \epsilon \leq \epsilon_1 ~\mbox{,}
\end{equation}
where
\begin{equation} \label{eq_defgamma}
\gamma := \int_{\mathbb{T}} \log \vert E - v(x + i \epsilon_1) \vert \ud x - \log(2) - \log \vert \lambda_M \vert - 2 \pi M \epsilon_1 ~\mbox{.}
\end{equation}
In particular, letting $\epsilon = 0$, one has
\begin{equation} \label{eq_imrpovedHerman}
L(\alpha, B^E) \geq \log|\lambda_M|+\gamma \frac{\epsilon_H}{\epsilon_H-\epsilon_1} ~\mbox{.}
\end{equation}
\end{prop}
\begin{remark}
Using Jensen's formula, the integral $\int_{\mathbb{T}} \log \vert E - v(x + i \epsilon_1) \vert \ud x$ can be evaluated based on the zeros of the polynomial 
\begin{equation}
f_{E,\epsilon_1}(z):=Ez^M-\sum_{j=1}^M(\lambda_jz^{j+M}e^{-2\pi j\epsilon_1}+\overline{\lambda_j} z^{-j+M}e^{2\pi j\epsilon_1}) ~\mbox{.}
\end{equation} 
Letting $a_1,\ldots,a_n$ be the zeros of $f_{E,\epsilon_1}(z)$ (counted with multiplicity) in the complex unit disk $\mathbb{D}$, $\gamma$ is given by
\begin{equation} \label{eq_gammafromjensen}
\gamma=\sum_{k=1}^n\log\frac{1}{|a_k|}-\log2 ~\mbox{.}
\end{equation}
\end{remark}

\begin{proof}
Consider the line segment connecting the points $(\epsilon_1,L(\epsilon_1;E))$ and $(\epsilon_0,L(\epsilon_0;E))$. 

By convexity of $\epsilon \mapsto L(\epsilon; E)$, one necessarily has 
\begin{equation}
\omega(\epsilon_1; E) \leq \frac{L(\epsilon_H;E)-L(\epsilon_1;E)}{2\pi(\epsilon_H-\epsilon_1)} =:a ~\mbox{,}
\end{equation}
whence
\begin{equation} \label{eq_supercrit_1}
L(\epsilon;E)\geq L(\epsilon_1;E)+2\pi a (\epsilon-\epsilon_1) ~\mbox{, $0 \leq \epsilon \leq \epsilon_1$} ~\mbox{.}
\end{equation} 

On the other hand, one has $L(\epsilon_H;E)=\log|\lambda_M|+2\pi M\epsilon_H$, and the lower bound (\ref{eq_complexLE_ds_lowerboundsupercrit}) from Proposition \ref{prop_key_amended} implies
\begin{equation} \label{eq_supercrit_2}
L(\epsilon_1; E) \geq \int_{\mathbb{T}} \log \vert E - v(x + i \epsilon_1) \vert \ud x - \log(2) = \gamma + \log \vert \lambda_M \vert + 2 \pi M \epsilon_1 ~\mbox{,}
\end{equation}
where we have made use of the definition of $\gamma$ in (\ref{eq_defgamma}).

Thus, we can estimate
\begin{equation}
a = \frac{L(\epsilon_H;E)-L(\epsilon_1;E)}{2\pi(\epsilon_H-\epsilon_1)} \leq M - \dfrac{\gamma}{2 \pi (\epsilon_H - \epsilon)} ~\mbox{,}
\end{equation}
which, combined with (\ref{eq_supercrit_1}), yields (\ref{eqn:LE-DS-ineq-ii}).
\end{proof}

\subsection{Example (numerics)}
Consider $v(x)=18\cos(2\pi x)+1.6\cos(4\pi x)$, so $\lambda_1=9$ and $\lambda_2=0.8$. Using Mathematica, we first computed numerically $\widetilde{m}(\epsilon; E) = \min_{x \in \mathbb{T}} \vert v(x + i \epsilon) \vert$, the results of which are shown in Fig. \ref{fig:min-plotting98}. 
\begin{figure}[h!]
\includegraphics[scale=0.5]{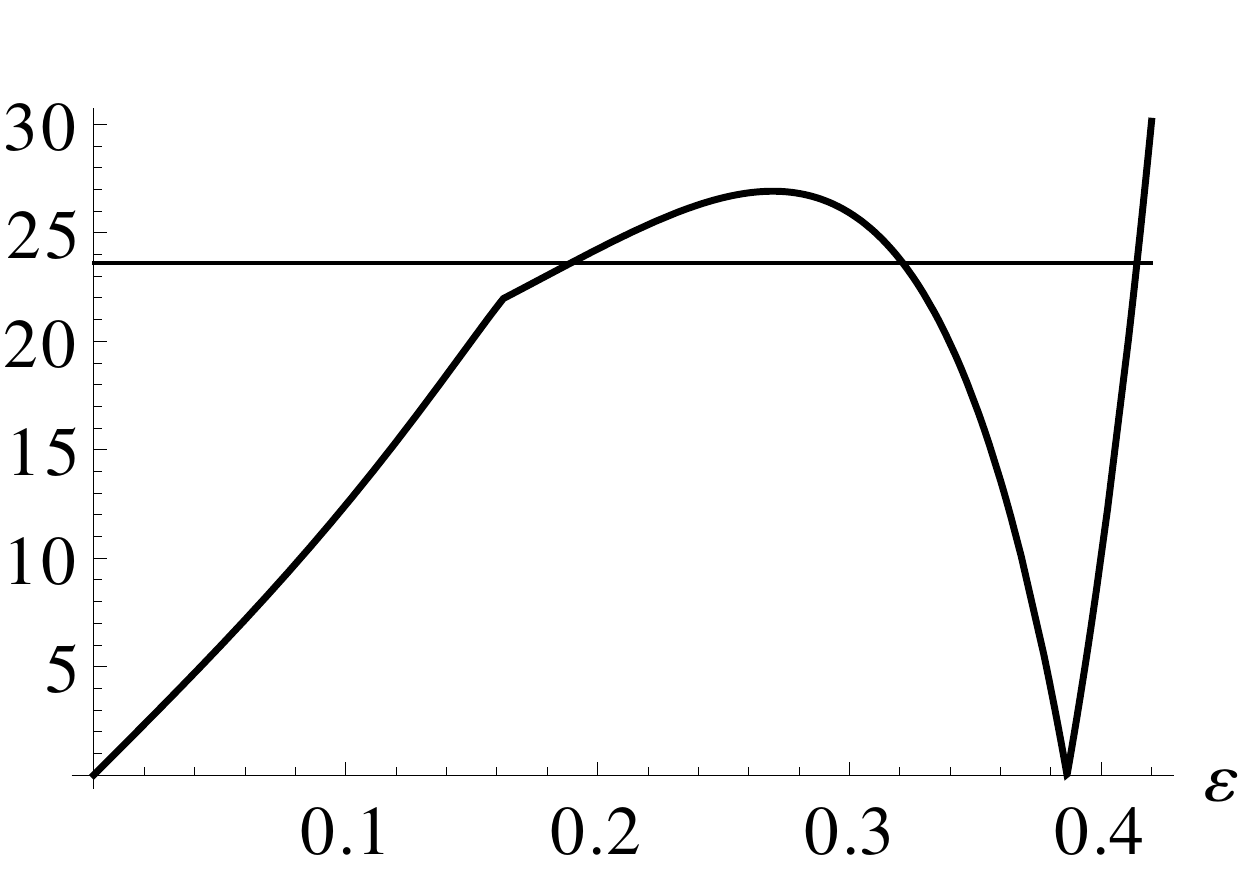}
\caption{Plot of $\widetilde{m}(\epsilon; E)$ for $\lambda_1=9$, $\lambda_2=0.8$. The horizontal line is drawn at $4+2\lambda_1+2\lambda_2=23.6$.}
\label{fig:min-plotting98}
\end{figure}

At $\epsilon=0.2$, $\min_{x\in\mathbb{T}}|v(x+i0.2)|\approx24.242>23.6=4+2\lambda_1+2\lambda_2$, so that $m(\epsilon; E)>2$ for all $E\in[-21.6,21.6]\supseteq\Sigma$; thus for any $E$ in this interval, we can apply Proposition \ref{prop_positiveLE} with $\epsilon_1 = 0.2$. 

For a specific example, take $E=-2$. The roots of $f(z)=-2z^2-9e^{-2\pi(0.2)}z^3-9e^{2\pi(0.2)}z-1.6e^{-4\pi(0.2)}z^4-1.6e^{4\pi(0.2)}$ are 
$-39.055, -0.3161$, and $-0.07839\pm 3.51271i$.
Only one root, $z=-0.3161$ is in $B_1(0)$, so from (\ref{eq_gammafromjensen}),
\begin{equation}
\gamma=-\log(0.3161)-\log 2=0.458.
\end{equation}

To apply (\ref{eq_imrpovedHerman}), we first solve numerically solve for the Herman radius, which for $E=-2$, yields $\epsilon_H=0.3864$. Then, the lower bound in (\ref{eq_imrpovedHerman}) implies
\begin{equation}
L(\alpha,B^E)=L(0;E=-2)\geq \log(0.8)+0.458\left(\frac{.3864}{.1864}\right)=0.727 > 0~\mbox{.}
\end{equation}
In comparison, the classical Herman bound gives $L(0;E)=L(\alpha,B^E)\geq \log|\lambda_M|\approx-0.223$, for all $E \in \mathbb{R}$.

Using Mathematica, we also sampled energies $E\in[-21.6,21.6] \supseteq \Sigma$, using a step size of $0.001$. We simplified the computation, by computing the uniform Herman radius instead of $\epsilon_H$ for each energy. From the proof of Proposition \ref{prop_positiveLE}, it is clear that (\ref{eqn:LE-DS-ineq-ii}) also holds for $\epsilon_H$ replaced by $\epsilon_{H; \mathrm{unif}}$, since then still $L(\epsilon_{H; \mathrm{unif}}; E) = \log \vert \lambda_M \vert + 2 \pi \epsilon_{H; \mathrm{unif}} M$ which was used in (\ref{eq_supercrit_2}). The computation leading to Fig.  \ref{fig:min-plotting98} allows to extract the uniform Herman radius, $\epsilon_{H; \mathrm{unif}} =0.4142$.

Applying the bound in (\ref{eq_imrpovedHerman}), numerical computation of $\gamma$ using (\ref{eq_gammafromjensen}), results 
\begin{equation}
L(0;E)\ge\log|\lambda_2|+\gamma \frac{\epsilon_{H; \mathrm{unif}}}{\epsilon_{H; \mathrm{unif}}-\varepsilon_1} =: L^-(E)>0 ~\mbox{, for all $E\in[-21.6,21.6] \subseteq \Sigma$.} 
\end{equation}
We show a plot of $L^-(E)$ in Fig. \ref{fig:example_energies}.
\begin{figure}[h!]
\includegraphics[width=2.7in]{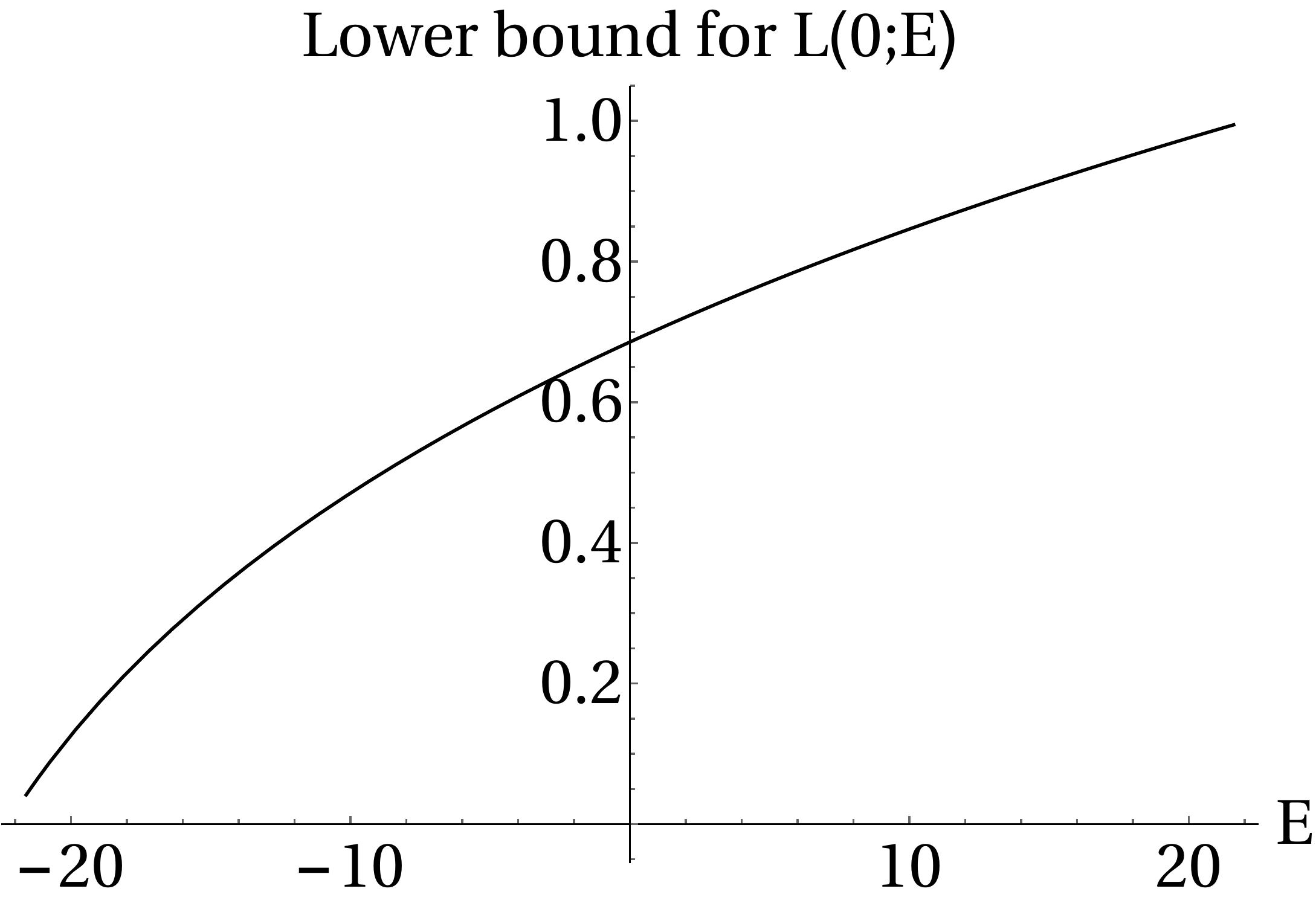}
\caption{Plot of the lower bound $L^-(E)$ for the LE extracted from Proposition \ref{prop_positiveLE} for $\lambda_1=9$, $\lambda_2=0.8$, and $E\in[-21.6,21.6] \supseteq \Sigma$.}
\label{fig:example_energies}
\end{figure}

\appendix

\section{Proof of Fact \ref{fact_symmetry}} \label{appendix_factsymmetry}

We approximate $\alpha$ by some rational $\frac{p}{q}$, in which case the resulting discrete Schr\"odinger operator becomes $q$-periodic. From the theory of periodic Schr\"odinger operators, it is know that for each $x \in \mathbb{T}$, the spectrum $\sigma(\frac{p}{q}, x)$ is determined by the {\em{discriminant}} $\Delta(\frac{p}{q}, x; E) = \mathrm{tr} \left\{ B_{q}^E(\frac{p}{q}; x)  \right\}$ via
\begin{equation}
\sigma(\frac{p}{q}, x) = \Delta(\frac{p}{q}, x; .)^{-1}([-2,2]) ~\mbox{.}
\end{equation}

We claim:
\begin{lemma} \label{lem_evenodd}
Let $x_0$ as in Fact \ref{fact_symmetry}, then for all $E \in \mathbb{R}$,
\begin{equation}
\Delta(\frac{p}{q}, x_0; E) = (-1)^q \Delta(\frac{p}{q}, x_0; -E) ~\mbox{.}
\end{equation}
In particular, if $q$ is odd, then $E=0 \in \sigma(\frac{p}{q}, x_0)$.
\end{lemma}
\begin{proof}
For simplicity denote $V_j = v(x_0 + j \frac{p}{q})$, $0 \leq j \leq q-1$, and, for $\mu \in \mathbb{R}$, set $A_\mu :=( \begin{smallmatrix}  \mu & -1 \\ 1& 0  \end{smallmatrix})$. Note that for every $\mu$,
\begin{equation}
A_\mu = (-1) A_{-\mu}^\mathrm{t} ~\mbox{,}
\end{equation}
where superscript $\mathrm{t}$ is the matrix transpose. Thus, using the anti-symmetry of $v(x_0 + .)$ and invariance of the trace under matrix transposition and cyclic permutation, we obtain
\begin{eqnarray}
\Delta(\frac{p}{q}, x_0; E) = & \mathrm{tr} \left\{ \left( \prod_{j=1}^{(q-1)/2} A_{E + V_j} \right) \left(  \prod_{j = (q-1)/2}^{1} A_{E - V_j} \right) A_E    \right\} \nonumber \\
      = & (-1)^q  \mathrm{tr} \left\{ \left( \prod_{j=1}^{(q-1)/2} A_{-E - V_j}^\mathrm{t} \right) \left(  \prod_{j = (q-1)/2}^{1} A_{-E + V_j}^\mathrm{t} \right) A_{-E}^\mathrm{t}    \right\} \nonumber \\
      = & (-1)^q \mathrm{tr} \left\{ A_{-E} \left( \prod_{j=1}^{(q-1)/2} A_{-E + V_j}    \right) \left(  \prod_{j = (q-1)/2}^{1} A_{-E - V_j}   \right) \right\} \nonumber \\
      = & (-1)^q \Delta(\frac{p}{q}, x_0; - E) ~\mbox{.}
\end{eqnarray}
\end{proof}

Finally, let $(\frac{p_n}{q_n})$ be any sequence of rationals with $\frac{p_n}{q_n} \to \alpha$ and $q_n$ odd for all $n \in \mathbb{N}$. Since the map $\mathbb{R} \ni \beta \mapsto \Sigma_+(\beta) := \cup_{x \in \mathbb{T}} \sigma(\beta , x)$ is known to be continuous in the Hausdorff metric at every irrational $\alpha$ \cite{AvronSimon_1983}, Lemma \ref{lem_evenodd} implies Fact \ref{fact_symmetry}.

\bibliographystyle{amsplain}

\end{document}